\documentclass[prl,twocolumn,aps,showpacs,footinbib,bibnotes]{revtex4-2}
\usepackage{graphicx}
\usepackage{physics,amsmath,amsfonts,amssymb,dsfont,amsthm}
\usepackage{mathtools}
\usepackage{multirow}
\usepackage[utf8]{inputenc}
\usepackage{comment}
\usepackage{ulem}
\theoremstyle{definition}
\newtheorem{theorem}{Theorem}

\newtheorem{prop}{Proposition}
\newtheorem{definition}{Definition}

\newtheorem{rem}{Remark}

\newcommand{\ie}{\textit{i}. \textit{e}.}
\newcommand{\sth}{\textit{s}. \textit{th}.}
\newcommand{\nr}[1]{{\color{red} \bf Note~:~#1}}

\usepackage{xcolor}
\definecolor{maroon}{RGB}{100,20,20}
\definecolor{dblue}{RGB}{20,20,100}
\usepackage[colorlinks=true,linkcolor=dblue,citecolor=maroon,urlcolor=maroon]{hyperref}
\begin{document}
\title{Quantum operations restricted by no faster-than-light
communication principle and generic emergence of objectivity
in position basis}
\author{Rajendra Singh Bhati}
\email{ph16076@iisermohali.ac.in}
\affiliation{Department of Physical Sciences, Indian
Institute of Science Education and Research (IISER) Mohali,
Sector 81 SAS Nagar, Manauli PO 140306 Punjab India}
\author{Arvind}
\email{arvind@iisermohali.ac.in}
\affiliation{Department of Physical Sciences, Indian
Institute of Science Education and Research (IISER) Mohali,
Sector 81 SAS Nagar, Manauli PO 140306 Punjab India}
\affiliation{Punjabi University, Patiala, 147002, Punjab india}

\begin{abstract}
The emergence of the objective classical world from the quantum
behavior of microscopic constituents is not fully understood.
Models based on decoherence and the principle of quantum
Darwinism, which attempt to provide such an explanation,
require system-bath interactions in a preferred basis. Thus
the generic emergence of objectivity in the position basis,
as observed in the real world remains unexplained.
In this Letter, we present a no-go theorem based
on the principle of no-faster-than-light communication,
showing that interactions between internal degrees of
freedom unavoidably cause system wave functions to branch
in the position basis. We
apply this result to
a spin decoherence model to demonstrate that a generic
thermal spin-$1/2$ bath redundantly records information
about the position of a spin-$1/2$
particle. Notably, the
model does not assume any preferred spin interaction. These
findings represent a compelling demonstration of the generic
emergence of objectivity in the position basis.
\end{abstract} 
\maketitle
\paragraph*{Introduction.---}
The principles of quantum theory imply that superposition of
distinct classical
possibilities is feasible at the atomic level. 
Further, the quantum states of composite systems in the
tensor product space involve superposition of classical
possibilities existing at  arbitrarily  large scales. Yet,
we do not directly witness large objects  exhibiting quantum
effects in our everyday observation of physical reality. 
Despite the universality of quantum theory, there
appears to be a quantum-to-classical transition taking place
when systems become macroscopic or they interact with
macroscopic systems and after this transition the system
behaves in a classically objective manner.
Physicists over a long
period have grappled with
the challenge of comprehending the quantum-to-classical
transition and the emergence of classical objectivity within the
framework of quantum theory. In recent decades,
substantial progress has been achieved in unraveling this
mystery and the concept of quantum
Darwinism has played a crucial role in in this process
~\cite{e_24111520,Zurek_2009,10.1063/PT.3.2550}.

Quantum Darwinism (QD) employs quantum
decoherence~\cite{RevModPhys_75.715,PhysRevLett_90.120404},
which is now a well-understood phenomenon, to explain the
emergence of objectivity. The central idea of QD is
that environment not only decoheres but also actively
records and proliferates the information about the system
attributes~\cite{PhysRevLett_93.220401,PhysRevA_72.042113,
PhysRevA_73.062310,PhysRevA.93.032126}.
The emergence of objectivity is identified with the
simultaneous maximization of the quantum mutual information
between an observable of the system and multiple
environment fragments~\cite{PhysRevLett_93.220401,PhysRevA_72.042113}.
Alternatively and equivalently, a specific
post-decoherence classical-quantum (CQ) state of the system
and environment, known as the spectrum broadcast structure,
can indicate
objectivity~\cite{PhysRevA_91.032122,PhysRevLett_118.120402,
Korbicz2021roadstoobjectivity,PhysRevLett.122.010403}.

While core predictions of QD are generic and thus are
independent of specific interaction
dynamics~\cite{Brandao2015,PhysRevLett.121.160401,Qi2021emergent},
there is a specific set of commuting observables, not always
easy to obtain, whose information gets  objectified by its
redundant recording on environment
fragments~\cite{PhysRevA.103.042210}. In this context,
classical objectivity has been shown to emerge in spin-spin
interaction
models~\cite{PhysRevLett_93.220401,PhysRevA_72.042113,Blume-Kohout_2005,
PhysRevLett_128.010401,Zwolak2016,PhysRevLett_112.140406},
dielectric illuminated
spheres~\cite{PhysRevLett_105.020404,PhysRevLett_112.120402,
Jess_2011,PhysRevA_91.032122} and quantum Brownian
motion~\cite{Tuziemski_2015,PhysRevLett_101.240405,
PhysRevA.80.042111,Tuziemski_2016,photonics2010228} with
ideal environments.  More practical decoherence models,
which relax various ideal conditions on the initial state of
the environment and the interaction Hamiltonian, have also
demonstrated
QD~\cite{PhysRevLett_103.110402,PhysRevA_81.062110,Zwolak_2016,
e_23111377,Jess_Riedel_2012,PhysRevA.99.042103,
PhysRevA_98.022124,RYAN2021127675,PhysRevResearch.2.012061,
PhysRevResearch.2.013164,PhysRevA.104.042216,PhysRevA.96.012120,
Balaneskovic2015,Balaneskovic2016}.  The predictions of QD
for certain models have also been experimentally
verified~\cite{GP2020,PhysRevLett.123.140402,CHEN2019580,
PhysRevA.98.020101,PhysRevLett.101.024102,PhysRevLett.104.176801}.

The aforementioned models have several limitations; they
exhibit preferred system observables in their interaction
dynamics, they do not capture the realistic universality
that the objectivity emerges in position bases and they do
not work with initial states and interaction dynamics that
are entirely random and generic.

This Letter presents a no-go theorem that imposes restricts
on interactions between two quantum
systems. More specifically, we demonstrate that
the manipulation of the internal
degrees of freedom of a system inevitably disrupts its
spatial wave function. These results
are proven using the no-faster-than-light communication
principle (NFLCP), which makes the theorem fundamental and
generic. Since, this no-go result
imposes particular constraints on
the interactions between quantum systems, we
utilize it to construct a generic decoherence model.
We observe that in this scenario the classical
objectivity emerges in the position basis irrespective of
the interaction dynamics of the internal degree of freedom.
Further, we make minimal assumptions about the initial
state of the environment and the interaction dynamics.

\paragraph*{No-go theorem.---} Before stating the no-go
result, let us define NFLCP in the information theoretic
framework.
\begin{definition}[NFLCP]
Let's suppose Alice generates a random bit string $A$ within
the space-time region $E_A$ and inputs it into a black box
$\mathbf{A}$. Similarly, Bob generates a bit string $B$ as
the output from a black box $\mathbf{B}$ within the
space-time region $E_B$. If $E_A$ and $E_B$ are separated by
a space-like interval, then $\mathcal{I}(A:B)=0$, where
$\mathcal{I}(A:B)$ represents the mutual information between
the strings $A$ and $B$.
\end{definition}
In non-relativistic quantum theory,
internal degrees of freedom of a particle are regarded
as independent physical systems. As a consequence,
while examining the dynamics of 
the spin of a
particle there is no necessity to
explicitly account for the spatial
wavefunction. Nonetheless, as we shall
see our
demonstration reveals the impossibility of consistently
describing dynamics without taking the spatial degree into
consideration.

\begin{figure}
\begin{center}
\includegraphics[scale=0.45]{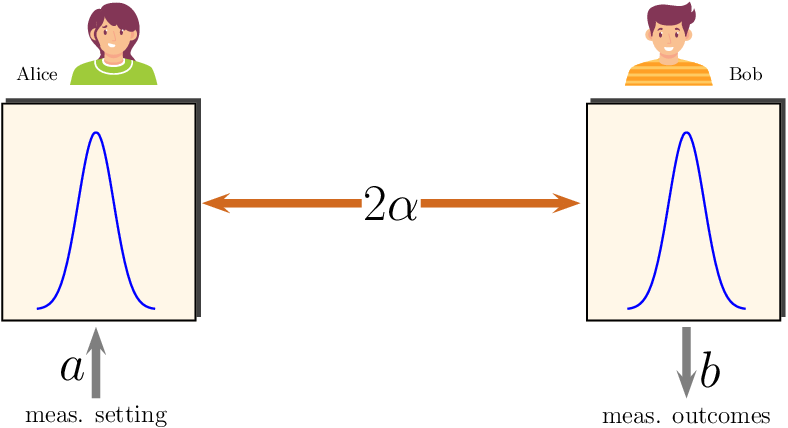}
\caption{Alice and Bob are stationed in two laboratories
separated by distance $2\alpha$. They share a common quantum
particle which is simultaneously present in both
laboratories. Alice chooses an operation based on randomly
generated bit $a$ and performs it on the internal degree of
freedom of the particle. Bob performs a fixed measurement on
the internal degree of freedom and registers the outcome as
a bit $b$.}
\label{nogo_fig}
\end{center}
\end{figure}

\begin{theorem}[No-go theorem]
\label{no-go_th}
If we have the Hilbert spaces $\mathcal{H}_S$ and
$\mathcal{H}_I$ associated with the spatial and internal
degrees of freedom of a quantum particle, respectively, then
the unitaries $U=\mathds{1}_S\otimes U_I$ and measurements
of the form
$\mathbb{M}\equiv\{\mathds{1}_S
\otimes\Pi_I,\mathds{1}_S\otimes\tilde{\Pi}_I\}$
on the state space $\mathcal{H}_S\otimes\mathcal{H}_I$ are
unallowed  by NFLCP. Here, $\Pi_I$ and $\tilde{\Pi}_I$ are
projection operators, and their sum satisfies
$\Pi_I+\tilde{\Pi}_I=\mathds{1}_I$.

\end{theorem}
\begin{proof}
We prove the theorem by contradiction.
Let us consider,
a spin half particle where the internal degree of freedom is a
two-level quantum system. Using a thought experiment, we
demonstrate that the operations mentioned in the theorem
violate the NFLCP.

Let us consider that the spin half particle is prepared in a state
where the spatial wavefunction extends over a large distance
and is a 
superposition of being located in the labs of observers
Alice and Bob, positioned at $x=-\alpha$ and
$x=\alpha$, respectively (see Fig.~\ref{nogo_fig}).
Additionally, let's assume that the  spin is prepared in the state
$\ket{0}$. Specifically, the composite state of the
particle, describing its spatial and spin degrees of
freedom, is given by
$\ket{\Psi}_{SI}=\ket{\psi}_S\otimes\ket{0}_I\in
\mathcal{H}_S\otimes\mathcal{H}_I$ such that
\begin{equation}
\braket{x}{\psi}=
N\left[\exp\left(-\frac{(x-\alpha)^2}{4\sigma^2}\right)+
\exp\left(-\frac{(x+\alpha)^2}{4\sigma^2}\right)\right]
\end{equation}
where $N$ is the normalization factor and $\sigma\ll\alpha$.
In this case, we make the assumption that the distance
between the two labs is
large in the sense that if $T$ is the typical operation
time for any process in the labs of Alice and Bob then $T
\ll \alpha/c$ (here $c$ is the velocity of light).

Given that the particle is spread across these two
labs, it effectively functions as a long black box that can
be accessed by observers in both the
labs. It is pertinent to note
that quantum mechanics does not attribute a notion of
`physical space' to the spin degree of freedom.
Consequently, it is reasonable to assume that the particle's
spin is accessible wherever its wavefunction is nonzero.
We use the scenario to prove the theorem.

Consider a situation where Alice wants to send  a classical bit of
information $a$  to Bob. 
Imagine a specific protocol where she applies an $a$
dependent 
operation $U_{a\in\{0,1\}}$ on the internal degree of
freedom of the partice in an initial state
$\ket{\Psi}_{SI}$.
Let us assume that
$U_0=\mathds{1}\otimes\mathds{1}$ and
$U_1=\mathds{1}\otimes\sigma_x$. 
Bob measures $\sigma_{z}$
and records the outcome as a bit $b\in\{0,1\}$.  Now assume
that $\tau$ is the time difference between Alice's bit
generation (call it event $E_A$) and Bob's act of recording
his measurement outcome (event $E_B$) \sth~ $c\tau\ll
2\alpha$. However, we
observe that $\mathcal{I}(A:B)=1$.  This contradicts NFLCP.

Similarly, consider another protocol for sending information
about the classical bit $a$ from Alice to Bob where Alice
performs $a$ dependent measurements on the particle.
Alice performs the
measurement $\mathbb{M}_{a\in{0,1}}$
with 
$\mathbb{M}_0\equiv\{\mathds{1}\otimes\mathds{1}\}$, meaning
she does not disturb the particle state, and
$\mathbb{M}_1\equiv\left\{\mathds{1}\otimes\ketbra{+},\mathds{1}\otimes\ketbra{-}\right\}$.
As before, Bob measures the complementary observable
$\sigma_z$ and records the outcome as a bit $b\in\{0,1\}$.
In this case, we find that $\mathcal{I}(A:B)=0.19$. Assuming
space-like separation of events of
measurements performed by Alice and Bob, this contradicts the NFLCP.

In this analysis, we have examined specific
instances of $U_I$ and $\Pi_I$. Nevertheless, it is apparent
that the obtained results can be extended to encompass
general operations.  
\end{proof}

A direct implication of Theorem~\ref{no-go_th} is that
quantum maps consistent with NFLCP must incorporate the
space-time region where the interactions take place. Let's
consider Alice's action of applying unitary $U$ or
measurement $\mathbb{M}$ as an event $(-\alpha,t_A)$ occuring
in the lab. In order to ensure the consistency of these
operations with NFLCP, we can express the unitary $U$
as $U=\ketbra{-\alpha}_S\otimes
U_I+(\mathds{1}-\ketbra{-\alpha})_S\otimes\mathds{1}_I$.
Similary the measurement operation $\mathbb{M}$,
occurring at her location, can be
represented as
$\mathbb{M}\equiv\{\ketbra{-\alpha}_S\otimes\Pi_I,
\ketbra{-\alpha}_S
\otimes\tilde{\Pi}_I,(\mathds{1}-\ketbra{-\alpha})_S\otimes\mathds{1}_I\}$. 
In both the cases the rationale being that these
operations are being performed in Alice's lab alone and there is
no information about these operations available outside her
lab.

The above scenario has an intriguing explanation in the
many-world interpretation
\cite{RevModPhys.29.454,DeWitt1970}: The ontology of the
world always branches in
the position basis, and transformations on internal degrees
happen accordingly. In the case considered above, for
instance, the unitary $U_I$ or the measurement
$\{\Pi,\tilde{\Pi}\}$ on $\mathcal{H}_I$ takes place in the
world where the particle is present at $-\alpha$. In all
other worlds, the internal degree remains unchanged.
Moreover, when arbitrary and repeated interactions occur
with an internal degree across different locations, it
results in a continuous branching in the position basis.
This branching process effectively resolves the preferred
basis
problem~\cite{Barrett2005-BARTPP,PhysRevD.24.1516,Inamori}.
Additionally, as we will explore further,
within an arbitrary spin environment model, the former leads
to the emergence of classical objectivity in the position
basis.

\paragraph*{The model.---}
A spin-$1/2$ quantum particle
denoted by
$\mathcal{S}$ is assumed to be in a spatial superposition of
being at $d$ possible locations $\{\vec{x}_1, \vec{x}_2,
\ldots, \vec{x}_d\} \equiv \mathbb{X}$. Let the initial
state of $\mathcal{S}$ be $\varrho_\mathcal{S} =
\ketbra{\Psi}_{\mathcal{S}} \otimes \rho_{\mathcal{S}} \in
\mathcal{H}_{\mathcal{S}}^x \otimes
\mathcal{H}_{\mathcal{S}}$, where $\ket{\Psi}_{\mathcal{S}}
= \sum_{i=1}^{d} \alpha_i \ket{\vec{x}_i}$ such that
$\sum_{i=1}^d\Vert\alpha_i\Vert^2=1$ represents the
spatial state and
$\rho_{\mathcal{S}}$ represents the spin state
of ${\mathcal{S}}$. Here,
$\mathcal{H}_{\mathcal{S}}^x$ and
$\mathcal{H}_{\mathcal{S}}$ are the Hilbert spaces
associated with the spatial and spin degrees of freedom of
the system.  The spin of the system $\mathcal{S}$ interacts
with an environment denoted as $\mathcal{E}$, which is
composed of point-like spin-$1/2$ particles fixed in
position. The interactions among environmental
subsystems (en-subs) are assumed to be absent. Let
$\mathcal{H}^x_{\mathcal{E}_i}$ and
$\mathcal{H}_{\mathcal{E}_i}$ denote the spatial and spin
state spaces of $i$-th en-sub $\mathcal{E}_i$. Initially,
each en-sub is in a random spin state and at a random
location: $\varrho_{\mathcal{E}_i} =
\ketbra{\vec{x}}_{\mathcal{E}_i} \otimes
\rho_{\mathcal{E}_i}\in\mathcal{H}_{\mathcal{E}_i}^x \otimes
\mathcal{H}_{\mathcal{E}_i}$, where $\vec{x} \in \mathbb{X}$
and $\rho_{\mathcal{E}_i}$ is an arbitrary spin state.
Furthermore, we assume that the spin-spin interaction
between ${\mathcal{S}}$ and $\mathcal{E}_i$ is arbitrary and
generic. The spin-spin interaction Hamiltonian in this
model, when the system and all en-subs are localized at
$\vec{x}\in\mathbb{X}$, is given as:

\begin{equation}
H_{{\mathcal{S}}:\mathcal{E}}(x)=
-\sum_{i=1}^N
g_i(x,t)\sigma_{\mathcal{S}}^i(x)\otimes\sigma_{\mathcal{E}_i}(x)
\bigotimes_{j\neq i}\mathds{1}_{\mathcal{E}_j},
\label{spin-spin}
\end{equation}
where $N$ is the number of en-subs,
$\sigma_{\mathcal{S}}^i(x)$ and $\sigma_{\mathcal{E}_i}(x)$
are random spin observables of ${\mathcal{S}}$ and
$\mathcal{E}_i$, respectively, when they interact at
position $\vec{x}$, and $g_i(x,t)$ is a function of time
that quantifies the interaction strength.  It is important
to note that, unlike all the previous spin models, we do not
assume any preferred spin observable for the system. In
fact, we have considered that an en-sub can couple to
different spin observable of the
system at different locations.

\begin{figure}
\begin{center}
\includegraphics[scale=0.38]{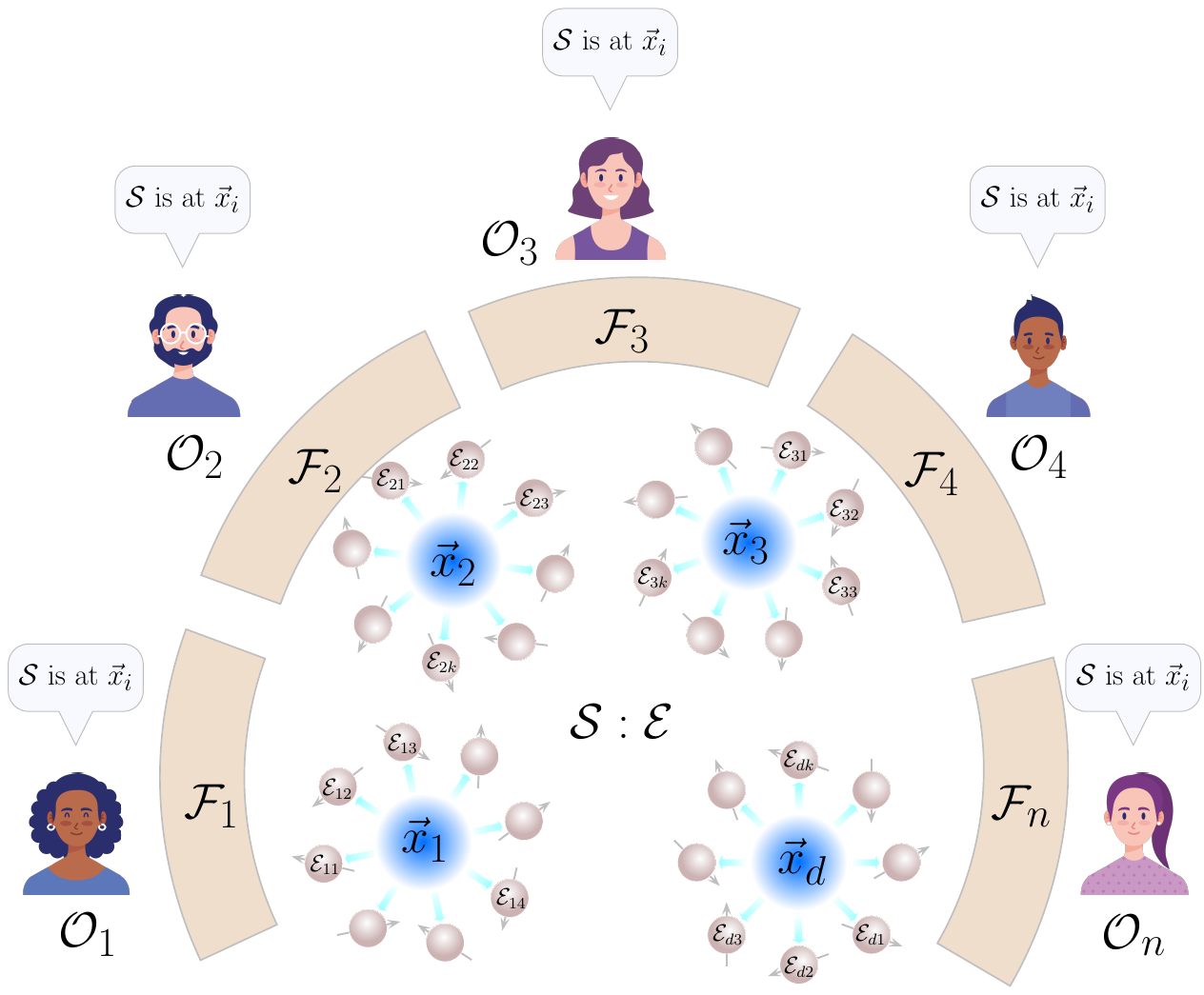}
\caption{The system $\mathcal{S}$ begins in a superposition
of different locations
$\vec{x}_1,\vec{x}_2,\cdots,\vec{x}_d$. The environment
subsystems (en-subs) $\{\mathcal{E}_{ij}\}$ are spin-$1/2$
particles localized at specific positions. En-subs near
$\vec{x}_i$ form the macro-fraction $\mathcal{E}_{i}^{mac}$.
$\mathcal{S}$ decoheres via spin-spin interaction with
en-subs. Fragments $\{\mathcal{F}_k\}$ consist of arbitrary
and uniformly random portion of the entire environment
$\mathcal{E}$, meaning $\mathcal{F}_k$ for all $k$ contains
en-subs from all macro-fractions in equal proportions.
Observers $\mathcal{O}_1,\mathcal{O}_2,\cdots,\mathcal{O}_n$
gain objective information about the position of
$\mathcal{S}$ after probing fragments
$\mathcal{F}_1,\mathcal{F}_2,\cdots,\mathcal{F}_n$,
respectively.}
\label{deco_fig}
\end{center}
\end{figure}

\paragraph*{Emergence of objectivity.---}
We demonstrate that the dynamics governed by
$H_{{\mathcal{S}}:\mathcal{E}}$ leads to the formation of a
spectrum broadcast structure in the system's position basis.

The initial state of the composite system is
$\varrho_{{\mathcal{S}}:\mathcal{E}} = \varrho_{\mathcal{S}}
\bigotimes_{i=1}^N \varrho_{\mathcal{E}_i}$.  Let us divide
$\mathcal{E}$ into $d$ macro-fractions
$\{\mathcal{E}^{mac}_k\}_{k\in\{1,2,\cdots,d\}}$ based on
the locations of the constituent en-subs
$\{\mathcal{E}_i\}_{i\in\{1,2,\cdots,N\}}$: all en-subs
located at $\vec{x}_k$ form the macro-fraction
$\mathcal{E}^{mac}_k$. Furthermore, let us re-index the
en-subs accordingly: the $l$-th en-sub in the $k$-th
macro-fraction $\mathcal{E}^{mac}_k$ is denoted by
$\mathcal{E}_{kl}$, where $l=1,2,\cdots,m_k$. Here, $m_k$
represents the total number of en-subs in the $k$-th
macro-fraction. See Fig.~\ref{deco_fig} for details.

Incorporating the spatial degree of freedom, the unitary
operator representing the interaction between
${\mathcal{S}}$ and $\mathcal{E}_{kl}$, in accordance with
the no-go result, can be formulated as:
\begin{equation}
\label{U_complete}
\begin{aligned}
\bar{U}_{{\mathcal{S}}:\mathcal{E}_{kl}}&=
\sum_{i=j}\ketbra{\vec{x}_i}_{\mathcal{S}}
\otimes\ketbra{\vec{x}_j}_{\mathcal{E}_{kl}}
\otimes{U_{{\mathcal{S}}:\mathcal{E}_{kl}}}(x_i)
\\
&+\sum_{i\neq j}\ketbra{\vec{x}_i}_{\mathcal{S}}
\otimes\ketbra{\vec{x}_j}_{\mathcal{E}_{kl}}
\otimes{\mathds{1}}_{{\mathcal{S}}:\mathcal{E}_{kl}}. \\
\end{aligned}
\end{equation}
Here, the spin-spin interaction of
${\mathcal{S}}:\mathcal{E}_{kl}$ at $\vec{x}_i$ is
calculated using
the interaction Hamiltoian described in Eq.(\ref{spin-spin})
and is represented by the unitary operator as:
\begin{equation}
U_{{\mathcal{S}}:\mathcal{E}_{kl}}(x_i)=
\exp{\iota\theta_{kl}(x_i)\sigma_{\mathcal{S}}^{kl}(x_i)\otimes\sigma_{\mathcal{E}_{kl}}(x_i)},
\end{equation}
where, $\theta_{kl}(x_i)=\int g_{kl}(x_i,t)dt$ is the
interaction strength.  We can make the assumption that the
en-subs interact with the system one-by-one, without
compromising generality. The complete interaction between
$\mathcal{S}$ and $\mathcal{E}$ is given by
$\bar{U}_{\mathcal{S}:\mathcal{E}}=
\prod_{i,j}\bar{U}_{{\mathcal{S}}:\mathcal{E}^{mac}_{ij}}\bigotimes_{k\neq
i,l\neq j}\mathds{1}_{\mathcal{E}^{mac}_{kl}}$. Let
$\mathcal{E}_{\setminus\tilde{\mathcal{E}}}$ represent the
environment after tracing out a subenvironment
$\tilde{\mathcal{E}}$ and, similarly, let
${\mathcal{E}^{mac}_{i\setminus\tilde{\mathcal{E}}_i}}$ be
the $i$-th macro-fraction after discarding a portion of it
that is denoted by $\tilde{\mathcal{E}}_i$. Tracing out
$\tilde{\mathcal{E}}_i$, for all $i$, decoheres the
composite system.  After tracing out the spatial degrees of
en-subs, the post-interaction state,
$\varrho^{\prime}_{{\mathcal{S}}:\mathcal{E}_{\setminus{\tilde{\mathcal{E}}}}}=
\Tr_{\tilde{\mathcal{E}}}\left(\bar{U}_{{\mathcal{S}}:
\mathcal{E}}\varrho_{{\mathcal{S}}:\mathcal{E}}
\bar{U}_{{\mathcal{S}}:\mathcal{E}}^\dagger\right)$,
for sufficiently large $\tilde{\mathcal{E}_i}\;\forall
i\in\{1,2,\cdots,d\}$ is obtained as (see \textit{supp. mat.}):
\begin{equation}
\label{broadcast_macro}
\varrho^{\prime}_{{\mathcal{S}}:\mathcal{E}_{\setminus\tilde{\mathcal{E}}}}
\approx\sum_{i=1}^{d}\Vert\alpha_i\Vert^2\ketbra{\vec{x}_i}_{\mathcal{S}}
\otimes\rho^{\prime}_{{\mathcal{S}}:{\mathcal{E}^{mac}_{i{\setminus\tilde{\mathcal{E}}_{i}}}}}
\bigotimes_{{j\neq i}}\rho_{\mathcal{E}^{mac}_{j{\setminus\tilde{\mathcal{E}}_{j}}}}.
\end{equation}
Here, portions $\{\tilde{\mathcal{E}_i}\}$ constitute the
whole of $\tilde{\mathcal{E}}$,
$\rho_{\mathcal{E}^{mac}_{j{\setminus\tilde{\mathcal{E}}_{j}}}}$
across all $j$ is the initial state of $j$-th macro-fraction
minus $\tilde{\mathcal{E}}_j$, and
\begin{equation}
\label{post_int_macro}
\begin{aligned}
\rho^{\prime}_{{\mathcal{S}}:\mathcal{E}^{mac}_{i\setminus\tilde{\mathcal{E}}_i}}
\approx\frac{1}{2}\mathds{1}_{\mathcal{S}}
\bigotimes_{j:\mathcal{E}_{ij}\in\mathcal{E}^{mac}_{i\setminus\tilde{\mathcal{E}}_i}}
\rho^\prime_{\mathcal{E}_{ij}}+\Omega_{{\mathcal{S}}:\mathcal{E}^{mac}_{i\setminus\tilde{\mathcal{E}}_i}}
\end{aligned}
\end{equation}
is the composite spin state of system and
$\mathcal{E}^{mac}_i$ after decoherence, where
$\rho^\prime_{\mathcal{E}_{ij}}=\cos^2(\theta_{ij})\rho_{\mathcal{E}_{ij}}
+\sin^2(\theta_{ij})\tilde{\rho}_{\mathcal{E}_{ij}}$
with
$\tilde{\rho}_{\mathcal{E}_{ij}}=\sigma_{\mathcal{E}_{ij}}
{\rho}_{\mathcal{E}_{ij}}\sigma^\dagger_{\mathcal{E}_{ij}}$
and the form of
$\Omega_{{\mathcal{S}}:\mathcal{E}^{mac}_{i\setminus\tilde{\mathcal{E}}_i}}$
is such that
$\Tr_{\mathcal{S}}\left(\Omega_{{\mathcal{S}}:
\mathcal{E}^{mac}_{i\setminus\tilde{\mathcal{E}}_i}}\right)=0$.

For clarity, the notation $x_i$ has been omitted from the
above expression for obvious reasons. Therefore, after
discarding the system's spin, the spin state of $i$-th
macro-fraction becomes
$\rho^{\prime}_{\mathcal{E}^{mac}_i}=\bigotimes_{j}\rho^\prime_{\mathcal{E}_{ij}}$. Index $j$ runs over the remaining en-subs of
$\mathcal{E}^{mac}_{i}$, that is,
$\mathcal{E}_{ij}\in\mathcal{E}^{mac}_{i\setminus\tilde{\mathcal{E}}_i}$.
Each en-sub in
$\mathcal{E}^{mac}_{i\setminus\tilde{\mathcal{E}}_{i}}$
contains a fraction of the information regarding the
system's presence (or absence) at $\vec{x}_i$.
However, it is worth noting that the macro-fraction
state
$\rho^{\prime}_{\mathcal{E}^{mac}_{i{\setminus\tilde{\mathcal{E}}_{i}}}}$,
when it is of sufficiently
large size, can contain redundant information
about the position.
Next we shall proceed to fragment the
environment in a manner that allows each fragment to carry
the complete classical information about the position of the
system.

Suppose there are $n$ observers, denoted as $\mathcal{O}_1,
\mathcal{O}_2, \dots, \mathcal{O}_n$, who have access to
different portions with equal proportions of each
macro-fraction
$\{\mathcal{E}^{mac}_i\}_{i\in\{1,2,\dots,d\}}$ within the
environment $\mathcal{E}$. We can represent these fragments
as $\mathcal{F}_1, \mathcal{F}_2, \dots, \mathcal{F}_n$ (see
Fig.~\ref{deco_fig}). The revised form of
Eq.~\eqref{broadcast_macro}, after tracing out the system's
spin, is as follows (see \textit{supp. mat.}):

\begin{equation}
\label{broadcast_fragments}
\varrho^{\prime}_{{\mathcal{S}}:\mathcal{E}_{\setminus\tilde{\mathcal{E}}}}=
\sum_{i=1}^{d}\Vert\alpha_i\Vert^2\ketbra{\vec{x}_i}_{\mathcal{S}}
\otimes{\xi^{\mathcal{F}_1}_i}\otimes{\xi^{\mathcal{F}_2}_i}
\otimes{\xi^{\mathcal{F}_3}_i}\otimes\cdots\otimes{\xi^{\mathcal{F}_n}_i}.
\end{equation}
Here, $\xi^{\mathcal{F}_k}_i$ represents the spin state of
the $k$-th fragment when the system is localized at
$\vec{x}_i$. This state, $\xi^{\mathcal{F}_k}_i$, is a
product state in which the en-subs present at $\vec{x}_i$
have the post-interaction spin state
$\rho^\prime_{\mathcal{E}_{ij}}$, while the remaining
en-subs are in their initial states
$\rho_{\mathcal{E}_{i^\prime j}}$ (here $i^\prime\neq i$).
The state $\varrho^\prime_{{\mathcal{S}}:\mathcal{E}}$,
Eq.~\eqref{broadcast_fragments}, is spectrum broadcast
structure if $\xi^{\mathcal{F}_k}_i$ and
$\xi^{\mathcal{F}_k}_{i^\prime}$ are perfectly
distinguishable for $i\neq i^\prime$ across all
$k\in\{1,2,\cdots, n\}$.

We use quantum fidelity measure to show the distinguishability
~\cite{UHLMANN1976273,Liang_2019}.
Since fidelity is multiplicative under tensor product~\cite{Jozsa1994}, we obtain:

\begin{equation}
F\left(\xi^{\mathcal{F}_k}_i,\xi^{\mathcal{F}_k}_{i^\prime}\right)=
\prod_{\substack{j\in\{i,i^\prime\}\\
j^\prime:\mathcal{E}_{jj^\prime}\in\mathcal{F}_k}}
F\left(\rho_{\mathcal{E}_{jj^{\prime}}},
\rho^{\prime}_{\mathcal{E}_{jj^\prime}}\right).
\end{equation}
Here, $F(\rho,\sigma)=\Tr\sqrt{\rho^{1/2}\sigma\rho^{1/2}}$
represents the fidelity between states $\rho$ and $\sigma$.
To clarify, the index $j^\prime$ spans all en-subs that
constitute the fragment $\mathcal{F}_k$, and also remember
that $i,i^\prime\in\{1,2,\cdots,d\}$. Since the interactions
are non-zero, meaning $\theta_{jj^\prime}\neq 0$, the
fidelity
$F\left(\rho_{\mathcal{E}_{jj^{\prime}}},
\rho^{\prime}_{\mathcal{E}_{jj^\prime}}\right)<1$.
Consequently, the product effectively approaches zero for
large fragments:
$F\left(\xi^{\mathcal{F}_k}_i,\xi^{\mathcal{F}_k}_{i^\prime}\right)\approx
0$ for $i\neq i^\prime$ across all $k\in\{1,2,\cdots, n\}$.
Thus, states $\xi^{\mathcal{F}_k}_i$ and
$\xi^{\mathcal{F}_k}_{i^\prime}$ become perfectly
distinguishable.


The structure described in Eq.~\eqref{broadcast_fragments}
represents a complete and redundant encoding of information
about the position of the system $S$ on the environment
spins. Observers gain complete information about
the position of the system when they access spins of
randomly sampled portions of the environment. However, it is necessary
that a significant portion of the environment which decoheres
the system's state is inaccessible. One intriguing aspect of
Eq.~\eqref{broadcast_fragments} is its inherent integration
of the Born probabilities $\{\Vert\alpha_i\Vert^2\}$ 
meaning
all the observers simultaneously observe the
system at $\vec{x}_i$ with probability $\Vert\alpha_i\Vert^2$.
The environment naturally selects, records and collapses the system's
wavefunction in the position basis as per Born rule in
Quantum Darwinian manner.

\paragraph*{Discussion.---}
In this Letter, we have presented an intriguing no-go result
and explored its implications for the generic emergence of
classical objectivity in the position basis. The
significance of the no-go result lies in its relevance to
quantum foundations: it demonstrates that interactions in
internal degrees of freedom are always mediated by the
spatial degree of freedom. Moreover, these interactions
couple the involved systems in the position basis,
effectively preventing
faster-than-light communications.  Our no-go theorem
directly challenges the widely held belief that internal
degrees of freedom can act as isolated physical systems that
can be entirely dissociated from spatial
wavefunctions of the involved particles~\cite{Aharonov_2013,Denkmayr2014}. 

This result bears implications for the
preferred basis problem for the models of emergence of
classical objectivity based on decoherence and QD.
Interactions among internal degrees, which are
pervasive at the atomic and sub-atomic levels, continually
lead to the unavoidable branching of the universe's
wavefunction in the position basis in the spirit of
many worlds interpretation of QM. As a consequence, the
position basis naturally emerges as the preferred
basis.

In our analysis, we
managed to successfully integrate the
constraints imposed by the no-go result into the dynamics of
a spin thermalization model, leading to the objectivity of
the system's spatial degree of freedom in the position
basis. A key strength of our model lies in its genuine
generality, as it avoids any reliance on a preferred spin
observable and preferred interaction Hamiltonian.
Another notable feature is the flexibility regarding the
environment's initial state, which can be arbitrary.
However, we do make certain assumptions: Firstly, we
consider non-interacting environmental spins, and secondly,
we set the self-Hamiltonian to zero. We acknowledge that
relaxing these assumptions requires further investigation
and should be a subject of future work.

Despite these advancements, there are important limitations
that need to be addressed. In our theorem and decoherence
model, we made simplifying assumptions of point-like
particles and interactions with vanishing range. To create a
more comprehensive and practical representation,
Eq.~\eqref{U_complete} should be extended to incorporate
interactions of non-vanishing range, while also considering
the speed limit of propagation for interaction influences.
Lastly, the no-go theorem where we have used NFLCP
within non-relativistic QM should be re-interpreted
and generalized within relativistic quantum mechanics.
Apart from its fundamental implications, such a result
can have potential importance in many-body dynamics and
energetics of quantum measurements.

Authors acknowledge the financial
support from {\sf DST/ICPS/QuST/Theme-1/2019/General}
Project number {\sf Q-68}.


%

\pagebreak





\clearpage
\newpage
\onecolumngrid
\begin{center}
	\textbf{\large Supplemental Materials:Quantum operations restricted by no faster-than-light communication principle and generic emergence of objectivity in the position basis}
\end{center}

\begin{center}
	{Rajendra Singh Bhati and Arvind}
\end{center}

\begin{center}
	\textit{Department of Physical Sciences, Indian
		Institute of Science Education and
		Research (IISER) Mohali, \\
		Sector 81 SAS Nagar, Manauli PO
		140306 Punjab India}
\end{center}
\setcounter{equation}{0}
\setcounter{figure}{0}
\setcounter{table}{0}
\setcounter{page}{1}
\makeatletter
\renewcommand{\theequation}{S\arabic{equation}}
\renewcommand{\thefigure}{S\arabic{figure}}
\renewcommand{\bibnumfmt}[1]{[S#1]}
\renewcommand{\citenumfont}[1]{S#1}

\section{Appendix--A: Detailed steps of the Proof of no-go theorem}

Here, we evaluate the mutual entropy $\mathcal{I}(A:B)$ in the thought experiment considered in the proof of the no-go theorem in the main text. We consider both cases one by one.

\begin{itemize}
	\item[(a)] Alice generates a uniformly random bit $a$ as a message to send it to Bob.
	The quantum state representing the bit can be written as
	\begin{equation}
		\label{proof_alice_bit}
		\rho_A=\frac{1}{2}\sum_{a\in\{0,1\}}\ketbra{a}_A.
	\end{equation}
	The classical-quantum (cq) state of Alice's bit and the spin-$1/2$ particle can be expressed as
	\begin{equation}
		\label{proof_cq}
		\rho_{ASI}=\frac{1}{2}\sum_{a\in\{0,1\}}\ketbra{a}_A\otimes\ketbra{\psi}_S\otimes\ketbra{0}_I.
	\end{equation}
	If Alice wants to send a bit $a\in\{0,1\}$ to Bob, she applies an operation $U_a$ on the particle's state
	$\ket{\Psi}_{SI}$. $U_a$ is specified as
	\begin{equation}
		U_a=\left\{
		\begin{array}{rll}
			\mathds{1}\otimes\mathds{1}; & \mbox{if}~ a=0, \\
			\mathds{1}\otimes{\sigma_x}; & \mbox{if}~ a=1. \\
		\end{array}
		\right.
	\end{equation}
	The state after Alice's operation becomes
	\begin{equation}
		\rho^\prime_{ASI}=\frac{1}{2}\sum_{a\in\{0,1\}}\ketbra{a}_A\otimes\ketbra{\psi}_S\otimes\ketbra{a}_I.
	\end{equation}
	Bob measures the operator $\sigma_{z}$ on the electron and records the outcome as a bit $b\in\{0,1\}$.
	After tracing out particle's wavefunction and spin state, we have classical-classical (cc) state of
	Alice and Bob's bits as
	\begin{equation}
		\label{nogo_cc_a}
		\rho_{AB}=\frac{1}{2}\sum_{a,b\in\{0,1\}}\ketbra{a}_A\otimes\ketbra{a}_B.
	\end{equation}
	The mutual entropy for $\rho_{AB}$ is $\mathcal{I}(A:B)=1$.
	
	\item[(b)] Similar to the previous case, Alice generates a bit $a\in\{0,1\}$, the state of which is given by Eq.~\eqref{proof_alice_bit} and the corresponding cq-state of Alice's bit and the particle is given
	by Eq.~\eqref{proof_cq}. To send bit $a=0$, Alice performs $\mathbb{M}_0\mathds{1}_S\otimes\mathds{1}_I$ on the
	particle, meaning she does not disturb its state. If the bit is $a=1$, she performs the following measurement:
	\begin{equation}
		\label{Alice_meas_nogo}
		\mathbb{M}_1=\left\{\mathds{1}\otimes\ketbra{+},\mathds{1}\otimes\ketbra{-}\right\}
	\end{equation}
	Notably, $\mathbb{M}_1$ is a measurement on the internal degree of freedom
	without disturbing the spatial wavefunction.
	The state after Alice's operation becomes
	\begin{equation}
		\begin{aligned}
			\rho^\prime_{ASI}&=\frac{1}{2}\ketbra{0}_A\otimes\ketbra{\psi}_S\otimes\ketbra{0}_I \\ &+\frac{1}{4}\ketbra{1}_A\otimes\ketbra{\psi}_S\otimes\sum_{k\in\{+,-\}}\ketbra{k}_I.
		\end{aligned}
	\end{equation}
	Bob measures $\sigma_{z}$ on the electron and records the outcome as a bit $b\in\{0,1\}$.
	After tracing out particle's wavefunction and spin state, we have classical-classical (cc) state of Alice and Bob's bits as
	\begin{equation}
		\begin{aligned}
			\label{nogo_cc_b}
			\rho_{AB}&=\frac{1}{2}\ketbra{0}_A\otimes\ketbra{0}_B+\frac{1}{4}\ketbra{1}_A\otimes\ketbra{1}_B \\ &+\frac{1}{4}\ketbra{1}_A\otimes\ketbra{0}_B \\
		\end{aligned}
	\end{equation}
	Using Eq.~\eqref{nogo_cc_b}, the mutual information between Alice and Bob is
	\begin{equation}
		\begin{aligned}
			\label{mutual_info_nogo_b}
			\mathcal{I}(A:B)&=1-h\left(\frac{1}{4}\right) \\
			&=1+\frac{1}{4}\log_2{\frac{1}{4}}+\frac{3}{4}\log_2{\frac{3}{4}} \\
			&\approx0.19 \\
		\end{aligned}
	\end{equation}
	where $h(\cdot)$ is the binary Shannon entropy.
\end{itemize}

Our main argument is based on the assumption that the spin of the electron (or the internal degree
of any quantum particle) is accessible at locations where ever the wavefunction is non-zero.
Moreover, we implicitly assume that the internal degree has no association with the spatial degree
of freedom and, thus, any manipulation at any point in space updates the spin-state at all points
in the space and that is how Alice and Bob are able to signal. In order to make all operations
spatially local, we need to include the notion of spatially localized quantum operations such as
\begin{equation}
	\label{U_actual_nogo}
	U^\prime=\ketbra{-\alpha}\otimes\sigma_x+(\mathds{1}-\ketbra{-\alpha})\otimes\mathds{1},
\end{equation}
or measurement of the form
\begin{equation}
	\begin{aligned}
		\label{M_actual_nogo}
		\mathbb{M}^\prime\equiv&\left\{\ketbra{-\alpha}\otimes\ketbra{+},\ketbra{-\alpha}\otimes\ketbra{-},\right. \\
		&\left.(\mathds{1}-\ketbra{-\alpha})\otimes\mathds{1}\right\}.
	\end{aligned}
\end{equation}
Eqs.~\eqref{U_actual_nogo} and \eqref{M_actual_nogo} incorporate the fact that manipulations on
spin that take place inside Alice's lab do not disturb the spin in Bob's lab. It is easy to follow that such
operations do not violate the no faster-than-light communication principle.
At first glance, Theorem~\ref{no-go_th} and its proof appear very trivial. However, our proof using
no faster-than-light communication principle highlights a deeper aspect of the connection between
the spatial wavefunction and the internal degree of freedom. Moreover, our theorem has established
that no manipulations (unitary or measurements) on the internal degree can be performed without
disturbing the spatial wavefunction.
If operations of the form $U$ or $\mathbb{M}$ (as specified in Theorem~\ref{no-go_th}) are not permitted,
it may be questioned what types of operations are permissible under the no faster-than-light
communication principle. An accurate answer to this question may not be plausible here. However,
we propose a possible solution which can be used in a crude way in certain physical scenarios to get
interesting results.

Here, we have assumed that the observer is sharply localized at a position $x$. This is an unrealistic
scenario. In a more practical situation, we can assume the effects of observer's action are reachable
in a spatial region $x\pm\delta$. In that case, we can replace the projection operator $\ketbra{x}$
by $\int_{x-\delta}^{x+\delta}\ketbra{x^\prime}dx^\prime$. So far, we have only considered the simple case of
one dimensional spatial degree of freedom. However, the generalization to three dimensional space
is straightforward and more realistic.  


\section{Appendix--B: Post-interaction state}

The post-interaction spin state of system and environment, $\varrho^\prime_{\mathcal{S}:\mathcal{E}}$, is evaluated as:
\begin{equation}
	\label{post_int_1}
	\varrho^\prime_{\mathcal{S}:\mathcal{E}}=\Tr_{\bigotimes_{ij}\mathcal{H}^x_{\mathcal{E}_{ij}}}\left(\bar{U}_{\mathcal{S}:\mathcal{E}}\varrho_{\mathcal{S}:\mathcal{E}}\bar{U}^\dagger_{\mathcal{S}:\mathcal{E}}\right).
\end{equation} 
Here, $\Tr_{\bigotimes_{ij}\mathcal{H}^x_{\mathcal{E}_{ij}}}(\cdot)$ represents tracing out spatial degrees of freedom of all environmental-subsystems (en-subs) $\{\mathcal{E}_{ij}:i=1,2,\cdots,d;\;\forall j\}$. The unitary $\bar{U}_{\mathcal{S}:\mathcal{E}}$ is given by
\begin{equation}
	\bar{U}_{\mathcal{S}:\mathcal{E}}=\prod_{k,l}\bar{U}_{{\mathcal{S}}:\mathcal{E}_{kl}}\bigotimes_{ij\neq kl}\mathds{1}_{\mathcal{E}_{ij}},
\end{equation}
where
\begin{equation}
	\label{U_complete_SM}
	\begin{aligned}
		\bar{U}_{{\mathcal{S}}:\mathcal{E}_{kl}}&=\sum_{i=j}\ketbra{\vec{x}_i}_{\mathcal{S}}\otimes\ketbra{\vec{x}_j}_{\mathcal{E}_{kl}}\otimes{U_{{\mathcal{S}}:\mathcal{E}_{kl}}}(x_i) \\
		&+\sum_{i\neq j}\ketbra{\vec{x}_i}_{\mathcal{S}}\otimes\ketbra{\vec{x}_j}_{\mathcal{E}_{kl}}\otimes{\mathds{1}}_{{\mathcal{S}}:\mathcal{E}_{kl}}. \\
	\end{aligned}
\end{equation}
The spin-spin interaction between the system $\mathcal{S}$ and en-sub $\mathcal{E}_{kl}\;\forall k,l$ at $\vec{x}_i$ is $U_{{\mathcal{S}}:\mathcal{E}_{kl}}(x_i)=\exp{\iota\theta_{kl}(x_i)\sigma_{\mathcal{S}}^{kl}(x_i)\otimes\sigma_{\mathcal{E}_{kl}}(x_i)}$. Notably, $\sigma_{\mathcal{S}}^{kl}(x_i)$ and $\sigma_{\mathcal{E}_{kl}}(x_i)$ are arbitrary spin observables. The variable $\theta_{kl}(x_i)$ represents the interaction strength. Let us now evaluate the post-interaction state $\varrho^\prime_{\mathcal{S}:\mathcal{E}}$. The initial state of the system and environment is expressed as:

\begin{equation}
	\label{initial_SE_mac}
	\begin{aligned}
		\varrho_{\mathcal{S}:\mathcal{E}}&=\left(\sum_{i,j}\alpha_i\alpha^{\ast}_j \ketbra{\vec{x}_i}{\vec{x}_j}_\mathcal{S} \otimes\rho_\mathcal{S}\right)\otimes\varrho_{\mathcal{E}^{mac}_1}\otimes\varrho_{\mathcal{E}^{mac}_2}\otimes\cdots\otimes\varrho_{\mathcal{E}^{mac}_d}. \\
		\\
	\end{aligned}
\end{equation}
Here, $\varrho_{\mathcal{E}^{mac}_i}=\bigotimes_{j}\varrho_{\mathcal{E}_{ij}}$.
Let us assume without loss of generality that the system interacts with macroscopic fractions one by one in the order $\mathcal{E}^{mac}_1,\mathcal{E}^{mac}_2,\cdots,\mathcal{E}^{mac}_d$. Furthermore, we can assume that en-subs interact with the system one by one within a macroscopic fraction. Suppose the order of interactions within the macro-fraction $\mathcal{E}^{mac}_i$ is $\mathcal{E}_{i1},\mathcal{E}_{i2},\mathcal{E}_{i3},\cdots,\mathcal{E}_{im_i}$. The corresponding unitary operation can be decomposed as:  
\begin{equation}
	\begin{aligned}
		&\bar{U}_{\mathcal{S}:\mathcal{E}^{mac}_{i}}=\bar{U}_{\mathcal{S}:\mathcal{E}_{i m_i}}\cdots\bar{U}_{\mathcal{S}:\mathcal{E}_{i3}}\bar{U}_{\mathcal{S}:\mathcal{E}_{i2}}\bar{U}_{\mathcal{S}:\mathcal{E}_{i1}}. \\
	\end{aligned}
\end{equation}
The interaction unitary $\bar{U}_{\mathcal{S}:\mathcal{E}_{11}}$ transforms the state $\varrho_{\mathcal{S}:\mathcal{E}}$ into $\varrho^{(11)}_{\mathcal{S}:\mathcal{E}}$, then $\bar{U}_{\mathcal{S}:\mathcal{E}_{12}}$ transforms $\varrho^{(11)}_{\mathcal{S}:\mathcal{E}}$ into $\varrho^{(12)}_{\mathcal{S}:\mathcal{E}}$ and so on. Let $\rho_{\mathcal{S}:\mathcal{E}^{mac}_{i}}$ denote the composite spin state of system and $i$-th macro-fraction \ie~$\rho_{\mathcal{S}:\mathcal{E}^{mac}_{i}}=\rho_{\mathcal{S}}\bigotimes_j\rho_{\mathcal{E}_{ij}}$. Hereafter, similar uses of this notation are understood. We obtain
\begin{equation}
	\begin{aligned}
		\varrho_{\mathcal{S}:\mathcal{E}}\quad\xrightarrow[]{\bar{U}_{\mathcal{S}:\mathcal{E}_{11}}} \quad &\varrho^{(11)}_{\mathcal{S}:\mathcal{E}} \\ =&\Vert\alpha_1\Vert^2\ketbra{\vec{x}_1}_{\mathcal{S}}\bigotimes_{i^\prime}\ketbra{\vec{x}_1}_{\mathcal{E}_{1i^\prime }}\otimes \left(U_{\mathcal{S}:\mathcal{E}_{11}}(x_1)\right)\rho_{\mathcal{S}:\mathcal{E}^{mac}_{1}}\left(U_{\mathcal{S}:\mathcal{E}_{11}}(x_1)\right)^{\dagger}\bigotimes_{j^\prime\neq 1}\varrho_{\mathcal{E}^{mac}_{j^\prime}} \\
		& +\left(\sum_{l\neq 1}\alpha_l\alpha_1^{\ast}\ketbra{\vec{x}_l}{\vec{x}_1}_{\mathcal{S}}\right)\bigotimes_{i^\prime}\ketbra{\vec{x}_1}_{\mathcal{E}_{1i^\prime}}\otimes \rho_{\mathcal{S}:\mathcal{E}^{mac}_{1}}\left(U_{\mathcal{S}:\mathcal{E}_{11}}(x_1)\right)^{\dagger}\bigotimes_{j^\prime\neq 1}\varrho_{\mathcal{E}^{mac}_{j^\prime}} \\
		& +\left(\sum_{m\neq 1}\alpha_1\alpha_m^{\ast}\ketbra{\vec{x}_1}{\vec{x}_m}_{\mathcal{S}}\right)\bigotimes_{i^\prime}\ketbra{\vec{x}_1}_{\mathcal{E}_{1i^\prime}}\otimes U_{\mathcal{S}:\mathcal{E}_{11}}(x_1)\rho_{\mathcal{S}:\mathcal{E}^{mac}_{1}}\bigotimes_{j^\prime\neq 1}\varrho_{\mathcal{E}^{mac}_{j^\prime}} \\
		& +\left(\sum_{l\neq 1,m\neq 1}\alpha_1\alpha_m^{\ast}\ketbra{\vec{x}_l}{\vec{x}_m}_{\mathcal{S}}\right)\bigotimes_{i^\prime}\ketbra{\vec{x}_1}_{\mathcal{E}_{1i^\prime}}\otimes \rho_{\mathcal{S}:\mathcal{E}^{mac}_{1}}\bigotimes_{j^\prime\neq 1}\varrho_{\mathcal{E}^{mac}_{j^\prime}}. \\
	\end{aligned}
\end{equation}
Similarly,
\begin{equation}
	\begin{aligned}
		\varrho^{(11)}_{\mathcal{S}:\mathcal{E}}\quad\xrightarrow[]{\bar{U}_{\mathcal{S}:\mathcal{E}_{12}}} \quad &\varrho^{(12)}_{\mathcal{S}:\mathcal{E}} \\ =&\Vert\alpha_1\Vert^2\ketbra{\vec{x}_1}_{\mathcal{S}}\bigotimes_{i^\prime}\ketbra{\vec{x}_1}_{\mathcal{E}_{1i^\prime }}\otimes \left(U_{\mathcal{S}:\mathcal{E}_{12}\mathcal{E}_{11}}(x_1)\right)\rho_{\mathcal{S}:\mathcal{E}^{mac}_{1}}\left(U_{\mathcal{S}:\mathcal{E}_{12}\mathcal{E}_{11}}(x_1)\right)^{\dagger}\bigotimes_{j^\prime\neq 1}\varrho_{\mathcal{E}^{mac}_{j^\prime}} \\
		& +\left(\sum_{l\neq 1}\alpha_l\alpha_1^{\ast}\ketbra{\vec{x}_l}{\vec{x}_1}_{\mathcal{S}}\right)\bigotimes_{i^\prime}\ketbra{\vec{x}_1}_{\mathcal{E}_{1i^\prime}}\otimes \rho_{\mathcal{S}:\mathcal{E}^{mac}_{1}}\left(U_{\mathcal{S}:\mathcal{E}_{12}\mathcal{E}_{11}}(x_1)\right)^{\dagger}\bigotimes_{j^\prime\neq 1}\varrho_{\mathcal{E}^{mac}_{j^\prime}} \\
		& +\left(\sum_{m\neq 1}\alpha_1\alpha_m^{\ast}\ketbra{\vec{x}_1}{\vec{x}_m}_{\mathcal{S}}\right)\bigotimes_{i^\prime}\ketbra{\vec{x}_1}_{\mathcal{E}_{1i^\prime}}\otimes U_{\mathcal{S}:\mathcal{E}_{12}\mathcal{E}_{11}}(x_1)\rho_{\mathcal{S}:\mathcal{E}^{mac}_{1}}\bigotimes_{j^\prime\neq 1}\varrho_{\mathcal{E}^{mac}_{j^\prime}} \\
		& +\left(\sum_{l\neq 1,m\neq 1}\alpha_1\alpha_m^{\ast}\ketbra{\vec{x}_l}{\vec{x}_m}_{\mathcal{S}}\right)\bigotimes_{i^\prime}\ketbra{\vec{x}_1}_{\mathcal{E}_{1i^\prime}}\otimes \rho_{\mathcal{S}:\mathcal{E}^{mac}_{1}}\bigotimes_{j^\prime\neq 1}\varrho_{\mathcal{E}^{mac}_{j^\prime}}, \\
	\end{aligned}
\end{equation}
where $U_{\mathcal{S}:\mathcal{E}_{12}\mathcal{E}_{11}}(x_1)=U_{\mathcal{S}:\mathcal{E}_{12}}(x_1)U_{\mathcal{S}:\mathcal{E}_{11}}(x_1)$. Continuing the above process of derivation, we obtain:
\begin{equation}
	\begin{aligned}
		\varrho_{\mathcal{S}:\mathcal{E}}\quad\xrightarrow[]{\bar{U}_{\mathcal{S}:\mathcal{E}^{mac}_{1}}} \quad &\varrho^{(1m_1)}_{\mathcal{S}:\mathcal{E}} \\ =&\Vert\alpha_1\Vert^2\ketbra{\vec{x}_1}_{\mathcal{S}}\bigotimes_{i^\prime}\ketbra{\vec{x}_1}_{\mathcal{E}_{1i^\prime }}\otimes \rho^{\prime}_{\mathcal{S}:\mathcal{E}^{mac}_{1}}\bigotimes_{j^\prime\neq 1}\varrho_{\mathcal{E}^{mac}_{j^\prime}} \\
		+& \left(\sum_{l\neq 1}\alpha_l\alpha_1^{\ast}\ketbra{\vec{x}_l}{\vec{x}_1}_{\mathcal{S}}\right)\bigotimes_{i^\prime}\ketbra{\vec{x}_1}_{\mathcal{E}_{1i^\prime}}\otimes \Xi^\dagger_{\mathcal{S}:\mathcal{E}_1^{mac}}\bigotimes_{j^\prime\neq 1}\varrho_{\mathcal{E}^{mac}_{j^\prime}} \\
		+& \left(\sum_{m\neq 1}\alpha_1\alpha_m^{\ast}\ketbra{\vec{x}_1}{\vec{x}_m}_{\mathcal{S}}\right)\bigotimes_{i^\prime}\ketbra{\vec{x}_1}_{\mathcal{E}_{1i^\prime}}\otimes \Xi_{\mathcal{S}:\mathcal{E}_1^{mac}}\bigotimes_{j^\prime\neq 1}\varrho_{\mathcal{E}^{mac}_{j^\prime}} \\
		+& \left(\sum_{l\neq 1,m\neq 1}\alpha_l\alpha_m^{\ast}\ketbra{\vec{x}_l}{\vec{x}_m}_{\mathcal{S}}\right)\bigotimes_{j^\prime}\varrho_{\mathcal{E}^{mac}_{j^\prime}}. \\
	\end{aligned}
\end{equation}
Here, we denote
\begin{equation} \label{rho_prime_notation}
	\begin{aligned}
		\rho^{\prime}_{\mathcal{S}:\mathcal{E}^{mac}_{i}} = & \left(U_{\mathcal{S}:\mathcal{E}_{im_i}\cdots\mathcal{E}_{i3}\mathcal{E}_{i2}\mathcal{E}_{i1}}(x_i)\right)\rho_{\mathcal{S}:\mathcal{E}^{mac}_{i}}\left(U_{\mathcal{S}:\mathcal{E}_{im_i}\cdots\mathcal{E}_{i3}\mathcal{E}_{i2}\mathcal{E}_{i1}}(x_i)\right)^\dagger \\
		\Xi_{\mathcal{S}:\mathcal{E}^{mac}_{i}} = & U_{\mathcal{S}:\mathcal{E}_{im_i}\cdots\mathcal{E}_{i3}\mathcal{E}_{i2}\mathcal{E}_{i1}}(x_i)\rho_{\mathcal{S}:\mathcal{E}^{mac}_{i}}. \\
	\end{aligned}
\end{equation}
The state $\varrho^{(2m_2)}_{\mathcal{S}:\mathcal{E}}=\left({\bar{U}_{\mathcal{S}:\mathcal{E}^{mac}_{2}}}\right)\varrho^{(1m_1)}_{\mathcal{S}:\mathcal{E}}\left({\bar{U}_{\mathcal{S}:\mathcal{E}^{mac}_{2}}}\right)^\dagger$ is evaluated in the similar manner:
\begin{equation}
	\begin{aligned}
		\varrho^{(1m_1)}_{\mathcal{S}:\mathcal{E}}\quad\xrightarrow[]{\bar{U}_{\mathcal{S}:\mathcal{E}^{mac}_{2}}} \quad &\varrho^{(2m_2)}_{\mathcal{S}:\mathcal{E}} \\ =&\Vert\alpha_1\Vert^2\ketbra{\vec{x}_1}_{\mathcal{S}}\bigotimes_{i^\prime}\ketbra{\vec{x}_1}_{\mathcal{E}_{1i^\prime }}\otimes \rho^{\prime}_{\mathcal{S}:\mathcal{E}^{mac}_{1}}\bigotimes_{j^\prime\neq 1}\varrho_{\mathcal{E}^{mac}_{j^\prime}} \\
		+&\Vert\alpha_2\Vert^2\ketbra{\vec{x}_2}_{\mathcal{S}}\bigotimes_{i^\prime}\ketbra{\vec{x}_2}_{\mathcal{E}_{2i^\prime }}\otimes \rho^{\prime}_{\mathcal{S}:\mathcal{E}^{mac}_{2}}\bigotimes_{j^\prime\neq 2}\varrho_{\mathcal{E}^{mac}_{j^\prime}} \\
		+&\alpha_1\alpha_2^\ast\ketbra{\vec{x}_1}{\vec{x}_2}_{\mathcal{S}}\bigotimes_{i^\prime}\ketbra{\vec{x}_1}_{\mathcal{E}_{1i^\prime }}\bigotimes_{i^\prime}\ketbra{\vec{x}_2}_{\mathcal{E}_{2i^\prime }}\otimes \Omega_{\mathcal{S}:\mathcal{E}^{mac}_{1}\mathcal{E}^{mac}_{2}}\bigotimes_{j^\prime\neq 1, 2}\varrho_{\mathcal{E}^{mac}_{j^\prime}} \\
		+&\alpha_1^\ast\alpha_2\ketbra{\vec{x}_2}{\vec{x}_1}_{\mathcal{S}}\bigotimes_{i^\prime}\ketbra{\vec{x}_1}_{\mathcal{E}_{1i^\prime }}\bigotimes_{i^\prime}\ketbra{\vec{x}_2}_{\mathcal{E}_{2i^\prime }}\otimes \Omega_{\mathcal{S}:\mathcal{E}^{mac}_{2}\mathcal{E}^{mac}_{1}}\bigotimes_{j^\prime\neq 1, 2}\varrho_{\mathcal{E}^{mac}_{j^\prime}} \\
		+&\left(\sum_{l\neq 1,2}\alpha_l\alpha_1^{\ast}\ketbra{\vec{x}_l}{\vec{x}_1}_{\mathcal{S}}\right)\bigotimes_{i^\prime}\ketbra{\vec{x}_1}_{\mathcal{E}_{1i^\prime}}\otimes \Xi^\dagger_{\mathcal{S}:\mathcal{E}_1^{mac}}\bigotimes_{j^\prime\neq 1}\varrho_{\mathcal{E}^{mac}_{j^\prime}} \\
		+&\left(\sum_{l\neq 1,2}\alpha_l\alpha_2^{\ast}\ketbra{\vec{x}_l}{\vec{x}_2}_{\mathcal{S}}\right)\bigotimes_{i^\prime}\ketbra{\vec{x}_2}_{\mathcal{E}_{2i^\prime}}\otimes \Xi^\dagger_{\mathcal{S}:\mathcal{E}_2^{mac}}\bigotimes_{j^\prime\neq 2}\varrho_{\mathcal{E}^{mac}_{j^\prime}} \\
		+&\left(\sum_{m\neq 1,2}\alpha_1\alpha_m^{\ast}\ketbra{\vec{x}_1}{\vec{x}_m}_{\mathcal{S}}\right)\bigotimes_{i^\prime}\ketbra{\vec{x}_1}_{\mathcal{E}_{1i^\prime}}\otimes \Xi_{\mathcal{S}:\mathcal{E}_1^{mac}}\bigotimes_{j^\prime\neq 1}\varrho_{\mathcal{E}^{mac}_{j^\prime}} \\
		+&\left(\sum_{m\neq 1,2}\alpha_2\alpha_m^{\ast}\ketbra{\vec{x}_2}{\vec{x}_m}_{\mathcal{S}}\right)\bigotimes_{i^\prime}\ketbra{\vec{x}_2}_{\mathcal{E}_{2i^\prime}}\otimes \Xi_{\mathcal{S}:\mathcal{E}_2^{mac}}\bigotimes_{j^\prime\neq 2}\varrho_{\mathcal{E}^{mac}_{j^\prime}} \\
		+& \left(\sum_{l\neq 1,2 \atop m\neq 1,2}\alpha_l\alpha_m^{\ast}\ketbra{\vec{x}_l}{\vec{x}_m}_{\mathcal{S}}\right)\bigotimes_{j^\prime}\varrho_{\mathcal{E}^{mac}_{j^\prime}}, \\
	\end{aligned}
\end{equation}
where we have denoted
\begin{equation}
	\begin{aligned}
		\Omega_{\mathcal{S}:\mathcal{E}^{mac}_{i}\mathcal{E}^{mac}_{j}} = & \left(U_{\mathcal{S}:\mathcal{E}_{im_i}\cdots\mathcal{E}_{i3}\mathcal{E}_{i2}\mathcal{E}_{i1}}(x_i)\right)\rho_{\mathcal{S}:\mathcal{E}^{mac}_{i}\mathcal{E}^{mac}_{j}}\left(U_{\mathcal{S}:\mathcal{E}_{jm_j}\cdots\mathcal{E}_{j3}\mathcal{E}_{j2}\mathcal{E}_{j1}}(x_j)\right)^\dagger. \\
	\end{aligned}
\end{equation}

It can be easily verified that the state $\varrho_{\mathcal{S}:\mathcal{E}}^{(dm_d)}=\bar{U}_{\mathcal{S}:\mathcal{E}}\varrho_{\mathcal{S}:\mathcal{E}}\bar{U}_{\mathcal{S}:\mathcal{E}}^\dagger$ is thus given by
\begin{equation}
	\begin{aligned}
		\varrho_{\mathcal{S}:\mathcal{E}}^{(dm_d)}=&\sum_{i=1}^d\Vert\alpha_i\Vert^2\ketbra{\vec{x}_i}_\mathcal{S}\bigotimes_{i^\prime}\ketbra{\vec{x}_i}_{\mathcal{E}_{ii^\prime}}\otimes\rho_{\mathcal{S}:\mathcal{E}_i^{mac}}^\prime\bigotimes_{j\neq i}\varrho_{\mathcal{E}_j^{mac}} \\
		+& \sum_{k\neq l}^d \alpha_k \alpha_l^\ast\ketbra{\vec{x}_k}{\vec{x}_l}_\mathcal{S}\bigotimes_{k^\prime}\ketbra{\vec{x}_k}_{\mathcal{E}_{kk^\prime}}\bigotimes_{l^\prime}\ketbra{\vec{x}_l}_{\mathcal{E}_{ll^\prime}}\otimes\Omega_{\mathcal{S}:\mathcal{E}_k^{mac}\mathcal{E}_l^{mac}}\bigotimes_{j\neq k,l}\varrho_{\mathcal{E}_j^{mac}} \\
	\end{aligned}
\end{equation}
After tracing out the spatial degree of en-subs, see Eq.~\eqref{post_int_1}, we obtain the composite spin state along with the spatial degree of the system,
\begin{equation}
	\begin{aligned}
		\label{post_int_2}
		\varrho_{\mathcal{S}:\mathcal{E}}^{\prime}=&\sum_{i=1}^d\Vert\alpha_i\Vert^2\ketbra{\vec{x}_i}_\mathcal{S}\otimes\rho_{\mathcal{S}:\mathcal{E}_i^{mac}}^\prime\bigotimes_{j\neq i}\rho_{\mathcal{E}_j^{mac}} \\
		+& \sum_{k\neq l}^d \alpha_k \alpha_l^\ast\ketbra{\vec{x}_k}{\vec{x}_l}_\mathcal{S}\otimes\Omega_{\mathcal{S}:\mathcal{E}_k^{mac}\mathcal{E}_l^{mac}}\bigotimes_{j\neq k,l}\rho_{\mathcal{E}_j^{mac}} \\
	\end{aligned}
\end{equation}

\section{Appendix--C: Spin state of the system and macro-fractions after decoherence}

Eq.~\eqref{post_int_2}, $\varrho_{\mathcal{S}:\mathcal{E}}^{\prime}$, represents a multipartite entangled state. Tracing out en-subs' spins decoheres $\varrho_{\mathcal{S}:\mathcal{E}}^{\prime}$. Here, we evaluate the state after tracing out significantly large portions of macro-fractions. Without loosing the generality, we assume that the first $n_i$ en-subs in $i$-th macro-fraction that interact with the system are discarded or cannot be accessed by observers. Since interactions among en-subs are assumed to be absent, we can evaluate effect of interaction with each en-sub separately. 
As we will show, every discarded en-sub strictly increases the mixness in the system spin.

\begin{prop}
	Let $\rho_\mathcal{S}=(\mathds{1}+\vec{r}\cdot\sigma)/2$, where $0\leq \Vert \vec{r}\Vert \leq 1$, and $\rho_{\mathcal{E}_{kl}}$ represent the density matrices corresponding to the spins of the system $\mathcal{S}$ and an en-sub $\mathcal{E}_{kl}$, respectively. Furthermore, let
	$\Omega=\begin{pmatrix}
	c_{11} & c_{12}\\
	c_{21} & c_{22} 
	\end{pmatrix}$
	be an arbitrary matrix (operator). Consider a unitary operator $U_{{\mathcal{S}}:\mathcal{E}_{kl}}=\exp{\iota\theta_{kl}\sigma_{\mathcal{S}}^{kl}\otimes\sigma_{\mathcal{E}_{kl}}}$ acting on the composite space $\mathcal{H}_{\mathcal{S}}\otimes\mathcal{H}_{\mathcal{E}_{kl}}$. Let $\rho_\mathcal{S}^\prime=(\mathds{1}+\vec{r^\prime}\cdot\sigma)/2$, where $0\leq \Vert \vec{r^\prime}\Vert \leq 1$, and $\Omega^\prime=\begin{pmatrix}
		c_{11}^\prime & c_{12}^\prime\\
		c_{21}^\prime & c_{22}^\prime 
	\end{pmatrix}$ be given as $\rho_\mathcal{S}^\prime=\Tr_{\mathcal{E}_{kl}}(U_{{\mathcal{S}}:\mathcal{E}_{kl}}\rho_{\mathcal{S}}\otimes\rho_{\mathcal{E}_{kl}} U^\dagger_{{\mathcal{S}}:\mathcal{E}_{kl}})$ and $\Omega^\prime=\Tr_{\mathcal{E}_{kl}}(U_{{\mathcal{S}}:\mathcal{E}_{kl}}\Omega\otimes\rho_{\mathcal{E}_{kl}})$, respectively. Then the followings are true:
	\begin{itemize}
		\item[(i)] $\Vert\vec{r^\prime}\Vert^2 = \Vert\vec{r}\Vert^2(1-\delta\sin[2](\theta_{kl})\sin[2](\phi))$,
		\item[(ii)] $\sum_{i,j\in\{1,2\}}\Vert c^\prime_{ij}\Vert^2 = \sum_{ij}\Vert c^\prime_{ij}\Vert^2\left(1-\delta\sin[2](\theta_{kl})\right)$,
	\end{itemize}
	where $\phi$ is the angle between $\sigma^{kl}_\mathcal{S}=\hat{s}\cdot\sigma$ and $\vec{r}=\Vert \vec{r}\Vert\hat{r}$ \ie~$\cos(\phi)=\hat{s}\cdot\hat{r}$, and $\delta=1-\left(\Tr{\sigma_{\mathcal{E}_{kl}}\rho_{\mathcal{E}_{kl}}}\right)^2$.
\end{prop}
\begin{proof}
\begin{itemize}
	\item[(i)] It is straightforward that
	\begin{equation}
		\label{s25}
		\rho^\prime_{\mathcal{S}}=\cos[2](\theta_{kl})\rho_{\mathcal{S}}+\sin[2](\theta_{kl})\tilde{\rho}_{\mathcal{S}}-\frac{r}{2}\sin(2\theta_{kl})\langle\sigma_{\mathcal{E}_{kl}}\rangle(\hat{s}\cdot\hat{r})\cdot\sigma,
	\end{equation}
	where $\tilde{\rho}_{\mathcal{S}}=\sigma^{kl}_{\mathcal{S}}\rho_{\mathcal{S}}(\sigma^{kl}_{\mathcal{S}})^\dagger$, $r=\Vert\vec{r}\Vert$, and $\langle\sigma_{\mathcal{E}_{kl}}\rangle=\Tr{\sigma_{\mathcal{E}_{kl}}\rho_{\mathcal{E}_{kl}}}$. Furthermore, we have
	\begin{equation}
		\begin{aligned}
			\tilde{\rho}_{\mathcal{S}}& = \frac{\mathds{1}}{2}+\frac{r}{2}\sigma^{kl}_{\mathcal{S}}(\hat{r}\cdot\sigma)(\sigma^{kl}_{\mathcal{S}})^\dagger \\
			& =\frac{\mathds{1}}{2} + \frac{r}{2}(\hat{s}\cdot\sigma)(\hat{r}\cdot\sigma)(\hat{s}\cdot\sigma) \\
			& =\frac{\mathds{1}}{2} + \frac{r}{2}(\hat{s}\cdot\hat{r}(\hat{s}\cdot\sigma)+i((\hat{s}\times\hat{r})\cdot\sigma)(\hat{s}\cdot\sigma)) \\
			& =\frac{\mathds{1}}{2} + \frac{r}{2}(\hat{s}\cdot\hat{r}(\hat{s}\cdot\sigma)-(\hat{s}\times\hat{r})\times\hat{s}\cdot\sigma) \\
			& =\frac{\mathds{1}}{2} + \frac{r}{2}\hat{k}\cdot\sigma, \\
		\end{aligned}
	\end{equation}
	where $\hat{k}=\cos(\phi)\hat{s}-\sin(\phi)\hat{n}$, $\hat{s}\cdot\hat{r}=\cos(\phi)$, $\hat{s}\times\hat{r}=\sin(\phi)\hat{n}^\prime$ and $\hat{n}^\prime\times\hat{s}=\hat{n}$. Note that $\hat{n}$ is a unit vector and so is $\hat{k}$. We can re-write Eq.~\eqref{s25} as
	\begin{equation}
		\begin{aligned}
			\rho^\prime_{\mathcal{S}}&=\cos[2](\theta_{kl})\left(\frac{\mathds{1}}{2} + \frac{r}{2}\hat{r}\cdot\sigma\right)+\sin[2](\theta_{kl})\left(\frac{\mathds{1}}{2} + \frac{r}{2}\hat{k}\cdot\sigma\right)-\frac{r}{2}\sin(2\theta_{kl})\sin(\phi)\langle\sigma_{\mathcal{E}_{kl}}\rangle\hat{n}^\prime\cdot\sigma \\
			& = \frac{1}{2}(\mathds{1}+ \vec{r^\prime}\cdot\sigma),
		\end{aligned}
	\end{equation}
	where
	\begin{equation}
		\vec{r^\prime}=r\left(\cos[2](\theta_{kl})\hat{r}+\sin[2](\theta_{kl})\hat{k}-\sin(2\theta_{kl})\sin(\phi)\langle\sigma_{\mathcal{E}_{kl}}\rangle\hat{n}^\prime\right).
	\end{equation}
	Since $\hat{r}=\cos(\phi)\hat{s}+\sin(\phi)\hat{n}$, thus
	\begin{equation}
		\Vert\vec{r^\prime}\Vert^2 = \Vert\vec{r}\Vert^2\left(1-\delta\sin[2](\theta_{kl})\sin[2](\phi)\right)
	\end{equation}

\item[(ii)] The operator $\Omega^\prime$ can be evaluated as
\begin{equation}
	\begin{aligned}
		\Omega^\prime&=\Tr_{\mathcal{E}_{kl}}(U_{{\mathcal{S}}:\mathcal{E}_{kl}}\Omega\otimes\rho_{\mathcal{E}_{kl}}) \\
		&=\cos(\theta_{kl})\Omega+\iota\langle{\sigma_{\mathcal{E}_{kl}}}\rangle\sin(\theta_{kl})\sigma^{kl}_{\mathcal{S}}\Omega.
	\end{aligned}
\end{equation}
Let us now assume the following general form for the observable $\sigma^{kl}_{\mathcal{S}}$:
\begin{equation}
	\sigma^{kl}_{\mathcal{S}}=\begin{pmatrix}
		\alpha & \beta e^{\iota\gamma}\\
		\beta e^{-\iota\gamma} & -\alpha\\
	\end{pmatrix},
\end{equation}
where $\alpha,\beta$ and $\gamma$ are real, and $\alpha^2+\beta^2=1$. Thus,
\begin{equation}
	\begin{aligned}
		&\Omega^\prime=\cos(\theta_{kl}){\begin{pmatrix}
				c_{11} & c_{12}\\
				c_{21} & c_{22}\\
		\end{pmatrix}}+\iota\langle{\sigma_{\mathcal{E}_{kl}}}\rangle\sin(\theta_{kl}){\begin{pmatrix}
				\alpha & \beta e^{\iota\gamma}\\
				\beta e^{-\iota\gamma} & -\alpha\\
		\end{pmatrix}}{\begin{pmatrix}
				c_{11} & c_{12}\\
				c_{21} & c_{22} 
		\end{pmatrix}} \\
		&=\begin{pmatrix}
			\cos(\theta_{kl})c_{11}+\iota\langle{\sigma_{\mathcal{E}_{kl}}}\rangle\sin(\theta_{kl})\left(\alpha c_{11}+\beta e^{\iota\gamma}c_{21}\right) & \cos(\theta_{kl})c_{12}+\iota\langle{\sigma_{\mathcal{E}_{kl}}}\rangle\sin(\theta_{kl})\left(\alpha c_{12}+\beta e^{\iota\gamma}c_{22}\right)\\
			\cos(\theta_{kl})c_{21}+\iota\langle{\sigma_{\mathcal{E}_{kl}}}\rangle\sin(\theta_{kl})\left(\beta e^{-\iota\gamma} c_{11}-\alpha c_{21}\right) & \cos(\theta_{kl})c_{22}+\iota\langle{\sigma_{\mathcal{E}_{kl}}}\rangle\sin(\theta_{kl})\left(\beta e^{-\iota\gamma} c_{12}-\alpha c_{22}\right) 
		\end{pmatrix} \\
		&\equiv {\begin{pmatrix}
				c_{11}^\prime & c_{12}^\prime\\
				c_{21}^\prime & c_{22}^\prime\\
		\end{pmatrix}}.
	\end{aligned}
\end{equation}
Furthermore, we have
\begin{equation}
	\begin{aligned}
		\Vert c^\prime_{11} \Vert^2&=\left(\cos[2](\theta_{kl})+\alpha^2\langle{\sigma_{\mathcal{E}_{kl}}}\rangle^2\sin[2](\theta_{kl})\right)\Vert c_{11}\Vert^2-\beta\langle{\sigma_{\mathcal{E}_{kl}}}\rangle\sin(2\theta_{kl})\Im{e^{\iota\gamma}c_{21}c_{11}^{\ast}} \\
		&+2\alpha\beta\langle{\sigma_{\mathcal{E}_{kl}}}\rangle^2\sin[2](\theta_{kl})\Re{e^{\iota\gamma}c_{21}c_{11}^{\ast}} + \beta^2\langle{\sigma_{\mathcal{E}_{kl}}}\rangle^2\sin[2](\theta_{kl})\Vert c_{21}\Vert^2, \\
		\Vert c^\prime_{12} \Vert^2&=\left(\cos[2](\theta_{kl})+\alpha^2\langle{\sigma_{\mathcal{E}_{kl}}}\rangle^2\sin[2](\theta_{kl})\right)\Vert c_{12}\Vert^2-\beta\langle{\sigma_{\mathcal{E}_{kl}}}\rangle\sin(2\theta_{kl})\Im{e^{\iota\gamma}c_{22}c_{12}^{\ast}} \\
		&+2\alpha\beta\langle{\sigma_{\mathcal{E}_{kl}}}\rangle^2\sin[2](\theta_{kl})\Re{e^{\iota\gamma}c_{22}c_{12}^{\ast}} + \beta^2\langle{\sigma_{\mathcal{E}_{kl}}}\rangle^2\sin[2](\theta_{kl})\Vert c_{22}\Vert^2, \\
		\Vert c^\prime_{21} \Vert^2&=\left(\cos[2](\theta_{kl})+\alpha^2\langle{\sigma_{\mathcal{E}_{kl}}}\rangle^2\sin[2](\theta_{kl})\right)\Vert c_{21}\Vert^2-\beta\langle{\sigma_{\mathcal{E}_{kl}}}\rangle\sin(2\theta_{kl})\Im{e^{-\iota\gamma}c_{11}c_{21}^{\ast}} \\
		&-2\alpha\beta\langle{\sigma_{\mathcal{E}_{kl}}}\rangle^2\sin[2](\theta_{kl})\Re{e^{-\iota\gamma}c_{11}c_{21}^{\ast}} + \beta^2\langle{\sigma_{\mathcal{E}_{kl}}}\rangle^2\sin[2](\theta_{kl})\Vert c_{11}\Vert^2, \\
		\Vert c^\prime_{22} \Vert^2&=\left(\cos[2](\theta_{kl})+\alpha^2\langle{\sigma_{\mathcal{E}_{kl}}}\rangle^2\sin[2](\theta_{kl})\right)\Vert c_{22}\Vert^2-\beta\langle{\sigma_{\mathcal{E}_{kl}}}\rangle\sin(2\theta_{kl})\Im{e^{-\iota\gamma}c_{12}c_{22}^{\ast}} \\
		&-2\alpha\beta\langle{\sigma_{\mathcal{E}_{kl}}}\rangle^2\sin[2](\theta_{kl})\Re{e^{-\iota\gamma}c_{12}c_{22}^{\ast}} + \beta^2\langle{\sigma_{\mathcal{E}_{kl}}}\rangle^2\sin[2](\theta_{kl})\Vert c_{12}\Vert^2. \\
	\end{aligned}
\end{equation}
Therefore,
\begin{equation}
	\begin{aligned}
		\sum_{i,j\in\{1,2\}}\Vert c^\prime_{ij}\Vert^2 & = \sum_{i,j\in\{1,2\}}\Vert c_{ij}\Vert^2\left(\cos[2](\theta_{kl})+\alpha^2\langle{\sigma_{\mathcal{E}_{kl}}}\rangle^2\sin[2](\theta_{kl})+\beta^2\langle{\sigma_{\mathcal{E}_{kl}}}\rangle^2\sin[2](\theta_{kl})\right) \\
		& = \sum_{i,j\in\{1,2\}}\Vert c_{ij}\Vert^2\left(\cos[2](\theta_{kl})+\langle{\sigma_{\mathcal{E}_{kl}}}\rangle^2\sin[2](\theta_{kl})\right) \\
		& = \sum_{i,j\in\{1,2\}}\Vert c_{ij}\Vert^2\left(1-\delta\sin[2](\theta_{kl})\right) \\		
	\end{aligned}
\end{equation}
\end{itemize}
\end{proof}

\begin{rem}
	Following inferences can be drawn from the above proposition:
	\begin{itemize}
		\item[(i)] For $\delta\sin[2](\theta_{kl})\sin[2](\phi)\neq 0$, the Bloch vector of system's spin strictly decreases after the interaction with corresponding en-sub. The former requires simultaneous fulfillment of three conditions: (i) $\delta\neq 0$, or equivalently $\langle\sigma_{\mathcal{E}_{kl}}\rangle^2\neq 1$, \ie~the initial state of the interacting en-sub is not an eigen state of the coupling observable $\sigma_{\mathcal{E}_{kl}}$. (ii) The interaction between the system and the en-sub is non-zero \ie~$\sin(\theta_{kl})\neq 0$. (iii) The system observable $\sigma^{kl}_{\mathcal{S}}$ is not aligned with the system's Bloch vector $\vec{r}$ \ie~$\sin(\phi)\neq 0$. Given that these conditions are fulfilled, the state of system's spin gets more and more mixed as it interacts with the environment. Notably, for a sufficiently large portion $\tilde{\mathcal{E}}_i$ of the macro-fraction $\mathcal{E}_i$, one can have
		\begin{equation}
			\rho^{\prime}_{\mathcal{S}}=\Tr_{\tilde{\mathcal{E}}_i}\left(\rho^{\prime}_{\mathcal{S}:\tilde{\mathcal{E}}_{i}}\right)=\frac{\mathds{1}}{2}.
		\end{equation}
		Here we have borrowed notations from Eq.~\eqref{rho_prime_notation}.
		\item[(ii)] For repeated interactions satisfying the condition $\delta\sin[2](\theta_{kl})\neq 0$ \ie~$\langle\sigma_{\mathcal{E}_{kl}}\rangle^2\neq 1$ and $\sin(\theta_{kl})\neq 0$, matrix $\Omega$ approaches to zero.
		Consequently,
		\begin{equation}
			\Xi^{\prime}_{\mathcal{S}}=\Tr_{\tilde{\mathcal{E}}_i}\left(\Xi_{\mathcal{S}:\tilde{\mathcal{E}}_{i}}\right)=0
		\end{equation}
		Here we have borrowed notations from Eq.~\eqref{rho_prime_notation}.  
	\end{itemize}
\end{rem}

\begin{prop} After tracing out a significantly large portion $\tilde{\mathcal{E}}$ of $\rho^\prime_{\mathcal{S}:\mathcal{E}}$, as specified in Eq.~\eqref{post_int_2}, the state of the system and the remaining environment is given by,
	\begin{equation}
		\begin{aligned}
			\label{post_int_3}
			\varrho_{\mathcal{S}:\mathcal{E}}^{\prime}=&\sum_{i=1}^d\Vert\alpha_i\Vert^2\ketbra{\vec{x}_i}_\mathcal{S}\otimes{\frac{\mathds{1}_{\mathcal{S}}}{2}}\bigotimes\rho_{\mathcal{E}_{i\setminus\tilde{\mathcal{E}}_i}^{mac}}\bigotimes_{j\neq i}\rho_{\mathcal{E}_j^{mac}}. \\
		\end{aligned}
	\end{equation}
Here, $\rho_{\mathcal{E}^{mac}_{i\setminus\tilde{\mathcal{E}_i}}}$ is the composite spin state of the en-subs in macro-fraction $\mathcal{E}_i^{mac}$ which are accessible.
\end{prop}
\begin{proof}
	The proof trivially follows from the above remark.
\end{proof}


Suppose the spin of the system is initially in maximally mixed state 
\ie~$\rho_\mathcal{S}=\frac{\mathds{1}_{\mathcal{S}}}{2}$. Let $\rho^{(k)}_{\mathcal{S}:\mathcal{E}^{mac}_i}$
denote the state of system plus $i$-th macro-environment after interaction
with $k$ en-subs.
It is straightforward to derive that  
\begin{equation}
	\begin{aligned}
		\rho^{(1)}_{\mathcal{S}:{\mathcal{E}}^{mac}_i}&=\left(U_{\mathcal{S}:\mathcal{E}_{i1}}\right)\rho_{\mathcal{S}:\mathcal{E}^{mac}_i}\left(U_{\mathcal{S}:\mathcal{E}_{i1}}\right)^\dagger \\
		&=\left(\frac{\mathds{1}}{2}\otimes\rho^\prime_{\mathcal{E}_{i1}}+\Omega_{1}\right)_{\mathcal{S}:\mathcal{E}_{i1}}\bigotimes_{j\neq 1}\rho_{\mathcal{E}_{ij}}
	\end{aligned}
\end{equation}
where
\begin{equation}
	\begin{aligned}
		\rho^\prime_{\mathcal{E}_{i1}}&=\left(\cos^2(\theta_{i1})\rho_{\mathcal{E}_{i1}}+\sin^2(\theta_{i1})\sigma_{\mathcal{E}_{i1}}\rho_{\mathcal{E}_{i1}}\sigma_{\mathcal{E}_{i1}}\right)\\
		&\equiv\left(\cos^2(\theta_{i1})\rho_{\mathcal{E}_{i1}}+\sin^2(\theta_{i1})\tilde{\rho}_{\mathcal{E}_{i1}}\right),\\
		\Omega_{1}
		&=\iota\sin(2\theta_{i1})\frac{\sigma^{i1}_{\mathcal{S}}}{2}\otimes\frac{\sigma_{\mathcal{E}_{i1}}\rho_{\mathcal{E}_{i1}}-\rho_{\mathcal{E}_{i1}}\sigma_{\mathcal{E}_{i1}}}{2}. \\
	\end{aligned}
\end{equation}
Note that $\Tr{\rho^\prime_{\mathcal{E}_{i1}}}=1$ and $\Tr_{\mathcal{S}}\left(\left(\Omega_1\right)_{\mathcal{S}:\mathcal{E}_{i1}}\right)\equiv 0$.
Similarly,
\begin{equation}
	\begin{aligned}
		\rho^{(2)}_{\mathcal{S}:{\mathcal{E}}^{mac}_i}&=\left(U_{\mathcal{S}:\mathcal{E}_{i2}}\right)\rho^{(1)}_{\mathcal{S}:\mathcal{E}^{mac}_i}\left(U_{\mathcal{S}:\mathcal{E}_{i2}}\right)^\dagger \\
		&=\left(\frac{\mathds{1}}{2}\otimes\rho^\prime_{\mathcal{E}_{i1}}\otimes\rho^\prime_{\mathcal{E}_{i2}}+\Omega_{2}\right)_{\mathcal{S}:\mathcal{E}_{i1}\mathcal{E}_{i2}}\bigotimes_{j\neq 1,2}\rho_{\mathcal{E}_{ij}}
	\end{aligned}
\end{equation}
where
\begin{equation}
	\begin{aligned}
		\rho^\prime_{\mathcal{E}_{i2}}&=\left(\cos^2(\theta_{i2})\rho_{\mathcal{E}_{i2}}+\sin^2(\theta_{i2})\sigma_{\mathcal{E}_{i2}}\rho_{\mathcal{E}_{i2}}\sigma_{\mathcal{E}_{i2}}\right)\\
		&\equiv\left(\cos^2(\theta_{i2})\rho_{\mathcal{E}_{i2}}+\sin^2(\theta_{i2})\tilde{\rho}_{\mathcal{E}_{i2}}\right),\\
	\end{aligned}
\end{equation}

and $\left(\Omega_2\right)_{\mathcal{S}:\mathcal{E}_{i1}\mathcal{E}_{i2}}$ is a traceless operator. More specifically, we have
\begin{equation}
	\Tr_\mathcal{S}\left(\left(\Omega_2\right)_{\mathcal{S}:\mathcal{E}_{i1}\mathcal{E}_{i2}}\right)\equiv 0.
\end{equation}
Suppose $m^\prime_i$ is the number of accessible en-subs in the macro-fraction $\mathcal{E}_{i}^{mac}$. It is straightforward that
\begin{equation}
	\label{trace_out_state}
	\begin{aligned}
		\rho^{\prime}_{{\mathcal{S}}:\mathcal{E}^{mac}_{i\setminus\tilde{\mathcal{E}}_i}}
		\approx\frac{1}{2}\mathds{1}_{\mathcal{S}}
		\bigotimes_{j:\mathcal{E}_{ij}\in\mathcal{E}^{mac}_{i\setminus\tilde{\mathcal{E}}_i}}
		\rho^\prime_{\mathcal{E}_{ij}}+\Omega_{{\mathcal{S}}:\mathcal{E}^{mac}_{i\setminus\tilde{\mathcal{E}}_i}}
	\end{aligned}
\end{equation}
\sth~
\begin{equation}
	\Tr_\mathcal{S}\left(\Omega_{{\mathcal{S}}:\mathcal{E}^{mac}_{i\setminus\tilde{\mathcal{E}}_i}}\right)\equiv 0.
\end{equation}
and
\begin{equation}
	\label{post_int_spin}
	\begin{aligned}
		\rho^\prime_{\mathcal{E}_{ij}}
		&=\cos^2(\theta_{ij})\rho_{\mathcal{E}_{ij}}+\sin^2(\theta_{ij})\tilde{\rho}_{\mathcal{E}_{ij}},\\
	\end{aligned}
\end{equation}
where
\begin{equation}
	\tilde{\rho}_{\mathcal{E}_{ij}}=\sigma_{\mathcal{E}_{ij}}\rho_{\mathcal{E}_{ij}}\sigma_{\mathcal{E}_{ij}}.
\end{equation}
After tracing out system's spin, we obtain
\begin{equation}
	\label{rho_prime_mac}
	\rho^{\prime}_{\mathcal{E}^{mac}_{i\setminus\tilde{\mathcal{E}}_i}}=\bigotimes_{j:\mathcal{E}_{ij}\in\mathcal{E}^{mac}_{i\setminus\tilde{\mathcal{E}}_i}}
	\rho^\prime_{\mathcal{E}_{ij}}.
\end{equation}
Therefore,
using Eqs.~\eqref{post_int_3} and \eqref{rho_prime_mac}, the post-interaction 
state of the system's spatial degree of freedom and the spins of the accessible environment (the spin of the system is traced out) is given by:
\begin{equation}
	\label{broadcast_macrostate_1}
	\varrho^{\prime}_{{\mathcal{S}}:\mathcal{E}_{\setminus\tilde{\mathcal{E}}}}
	\approx\sum_{i=1}^{d}\Vert\alpha_i\Vert^2\ketbra{\vec{x}_i}_{\mathcal{S}}
	\otimes\rho^{\prime}_{{\mathcal{E}^{mac}_{i{\setminus\tilde{\mathcal{E}}_{i}}}}}
	\bigotimes_{{j\neq i}}\rho_{\mathcal{E}^{mac}_{j{\setminus\tilde{\mathcal{E}}_{j}}}}.
\end{equation}
\section{Appendix--D: Formation of the spectral broadcast structure}
Remember that the spatial degrees of freedom of en-subs and spin of the system
are traced out. Additionally, a fraction of the environment $\tilde{\mathcal{E}}$ which is inaccessible by the observers is also traced out. As we will see, Eq.~\eqref{broadcast_macrostate_1} is a spectrum broadcast structure
where the information about the system's position is redundantly imprinted on
multiple fragments of environment-spins. Since we have discarded the spatial degree of freedom
of all en-subs, our environment $\mathcal{E}$ consists of only en-sub spins hereafter.
Let us now divide $\mathcal{E}$ into fragments $\mathcal{F}_1,\mathcal{F}_2,\cdots,\mathcal{F}_n$ in such a way that $\mathcal{F}_k$ for all
$k\in\{1,2,\cdots,n\}$ has randomly
chosen en-subs from all macro-fractions $\{\mathcal{E}^{mac}_j\}_{j\in\{1,2,\cdots,d\}}$.
This can be achieved by applying a random permutation on all en-subs
in Eq.~\eqref{broadcast_macrostate_1} and then dividing them into $n$ fragments of equal size.
Let us now denote the post-interaction state of the environment corresponding to the system's
position $\vec{x}_i$ by
\begin{equation}
	\chi_i=\rho^{\prime}_{{\mathcal{E}^{mac}_{i{\setminus\tilde{\mathcal{E}}_{i}}}}}
	\bigotimes_{{j\neq i}}\rho_{\mathcal{E}^{mac}_{j{\setminus\tilde{\mathcal{E}}_{j}}}}.
\end{equation}
With re-indexing, the state $\chi_i$ can be rewritten as:
\begin{equation}
	\chi_i=\rho_{\mathcal{E}_{11}}\otimes\rho_{\mathcal{E}_{12}}\otimes\cdots\otimes
	\underbrace{\rho^{\prime}_{\mathcal{E}_{i1}}\otimes\rho^{\prime}_{\mathcal{E}_{i2}}\otimes\rho^{\prime}_{\mathcal{E}_{i3}}\cdots\otimes\rho^{\prime}_{\mathcal{E}_{im^\prime_i}}}_{\rho^{\prime}_{\mathcal{E}^{mac}_{i}}}
	\otimes\rho_{\mathcal{E}_{(i+1)1}}\otimes\rho_{\mathcal{E}_{(i+1)2}}\otimes\cdots,
\end{equation}
where $\rho_{\mathcal{E}kl}$ and $\rho^{\prime}_{\mathcal{E}_{kl}}$ are initial and post-interaction states of
the $kl$-th en-sub, respectively. Remember that states $\{\rho^{\prime}_{\mathcal{E}_{kl}}\}$ are given by
Eq.~\eqref{post_int_spin}.
After a random shuffling (permutation) and re-indexing on en-subs, we obtain
\begin{equation}
	\begin{aligned}
		\chi_i=&\underbrace{\rho_{\mathcal{E}_1}\otimes\rho^{\prime}_{\mathcal{E}_2}\otimes\rho_{\mathcal{E}_3}\cdots}_{\mathcal{F}_1}\otimes
		\underbrace{\rho^{\prime}_{{\mathcal{E}_{r_1+1}}}\otimes\rho^{\prime}_{\mathcal{E}_{r_1+2}}\otimes\rho_{\mathcal{E}_{r_1+1}}\cdots}_{\mathcal{F}_2}\otimes
		\underbrace{\rho^{\prime}_{\mathcal{E}_{r_1+r_2+1}}\otimes\rho_{\mathcal{E}_{r_1+r_2+2}}\otimes\rho^{\prime}_{r_1+r_2+3}\cdots}_{\mathcal{F}_3}
		\\ 
		&\otimes
		\underbrace{\rho^{\prime}_{\sum_k^{n-1} r_k+1}\otimes\rho_{\sum_k^{n-1} r_k+2}\otimes\rho_{\sum_k^{n-1} r_k+3}\cdots}_{\mathcal{F}_n} \\
		\equiv&{\xi^{\mathcal{F}_1}_i}\otimes{\xi^{\mathcal{F}_2}_i}\otimes{\xi^{\mathcal{F}_3}_i}\otimes\cdots
		\otimes{\xi^{\mathcal{F}_n}_i},\\
	\end{aligned}
\end{equation}
where $\xi^{\mathcal{F}_k}_i$ is the state of $k$-th fragment corresponding to the system's position
$\vec{x}_i$. Here, $\{\rho_{\mathcal{E}_k}\}$ and $\{\rho^\prime_{\mathcal{E}_l}\}$ are randomly sampled from $
\bigotimes_{{j\neq i}}\rho_{\mathcal{E}^{mac}_{j{\setminus\tilde{\mathcal{E}}_{j}}}}$
and $\rho^{\prime}_{{\mathcal{E}^{mac}_{i{\setminus\tilde{\mathcal{E}}_{i}}}}}$, respectively. Note that $\xi^{\mathcal{F}_k}_i$ has multiple perturbed (post-interaction)
and unperturbed (initial) en-subs in the product state. 
Eq.~\eqref{broadcast_macrostate_1} can now be re-expressed as:
\begin{equation}
	\label{final_broadcast_structure}
	\begin{aligned}
		\varrho^{\prime}_{\mathcal{S}:\mathcal{E}}= &\sum_{i=1}^{d}\Vert\alpha_i\Vert^2\ketbra{\vec{x}_i}_{\mathcal{S}}\otimes{\xi^{\mathcal{F}_1}_i}\otimes{\xi^{\mathcal{F}_2}_i}\otimes{\xi^{\mathcal{F}_3}_i}\otimes\cdots\otimes{\xi^{\mathcal{F}_n}_i}. \\
	\end{aligned}
\end{equation}
Let us now prove that states $\xi^{\mathcal{F}_k}_i$ are perfectly distinguishable \ie~
\begin{equation}
	\xi^{\mathcal{F}_k}_i\xi^{\mathcal{F}_k}_j=0\:\forall \: i\neq j, \quad k=1,2,3,\cdots,n,
\end{equation}
or equivalently, the fidelity of states $\xi^{\mathcal{F}_k}_i$ and $\xi^{\mathcal{F}_k}_j$ for $i\neq j$ is zero:
\begin{equation}
	F\left(\xi^{\mathcal{F}_k}_i,\xi^{\mathcal{F}_k}_j\right)=0
\end{equation}
The fidelity of two density matrices $\rho$ and $\sigma$ is defined as
\begin{equation}
	F\left(\rho,\sigma\right)=\Tr\sqrt{\rho^{1/2}\sigma\rho^{1/2}}.
\end{equation}
Further, the fidelity is multiplicative under tensor products \ie
\begin{equation}
	F\left(\rho_1\otimes\rho_2,\sigma_1\otimes\sigma_2\right)=F\left(\rho_1,\sigma_1\right)F\left(\rho_2,\sigma_2\right),
\end{equation}
Since the fidelity for same states is one \ie~$F(\rho,\rho)=1$, we obtain
\begin{equation}
	F\left(\xi^{\mathcal{F}_k}_i,\xi^{\mathcal{F}_k}_{i^\prime}\right)=
	\prod_{\substack{j\in\{i,i^\prime\}\\
			j^\prime:\mathcal{E}_{jj^\prime}\in\mathcal{F}_k}}
	F\left(\rho_{\mathcal{E}_{jj^{\prime}}},
	\rho^{\prime}_{\mathcal{E}_{jj^\prime}}\right).
\end{equation}
The fidelity corresponding to unperturbed en-subs within fragment $\mathcal{F}_k$ is one. However, for perturbed en-subs the fidelity is strictly less than one given that $\sin(\theta_{jj^\prime})\neq 0$. Therefore, in the asymptomatic case where the size of the environment is infinitely large, we have
\begin{equation}
	F\left(\xi^{\mathcal{F}_k}_i,\xi^{\mathcal{F}_k}_{i^\prime}\right)\approx0 \quad \forall \: i\neq i^\prime, \quad k=1,2,3,\cdots,n.
\end{equation}
This proves our main claim.


\end{document}